\documentclass[11pt, letterpaper]{article}

\usepackage{fixltx2e}
\usepackage{authblk}
\usepackage{graphicx}
\usepackage[dvipsnames]{xcolor}
\usepackage{xspace}
\usepackage{xparse}
\usepackage[american]{babel}

\usepackage[footnotesize,nooneline,bf]{caption}

\usepackage{mathtools}
\usepackage{ifxetex,ifluatex}
\ifxetex
\usepackage{xltxtra}
\else
\ifluatex
\usepackage{luatextra}
\else
\usepackage{amssymb}
\usepackage[cal=rsfso, scr=rsfso, bb=txof]{mathalfa}
\usepackage[utf8]{inputenc}
\usepackage[T1]{fontenc}
\usepackage{mathpazo,tgcursor,tgtermes}
\usepackage[scale=.9]{tgheros}
\usepackage{textcomp}
\csname else\expandafter\endcsname
\romannumeral -`0% no space till \relax
\fi
\fi
\iftrue\relax%
\defaultfontfeatures{Ligatures=TeX}
\usepackage[math-style=TeX, vargreek-shape=TeX]{unicode-math}
\AtBeginDocument{%
  \let\setminus\smallsetminus% no U+29F5 in TeX Gyre Termes Math?
}
\fi
\providecommand{\restriction}{\upharpoonright}
\providecommand{\mathbfscr}[1]{\boldsymbol{\mathscr{#1}}}

\usepackage{microtype}

\usepackage{amsthm}

\usepackage{enumitem}

\usepackage{tabulary}
\usepackage{multirow}

\usepackage{tikz}
\usetikzlibrary{positioning, scopes,
  decorations.pathreplacing, calc,
  shapes.multipart}
 
\newcommand{\tikzmark}[1]{\tikz[overlay,remember picture] \node (#1) {};}

\usepackage[pagebackref,bookmarksnumbered,bookmarksopen,plainpages=false,pdfpagelabels,%
unicode,
breaklinks,colorlinks,citecolor=blue,linkcolor=blue,hyperindex]{hyperref}
\usepackage[noadjust]{cite}

\usepackage[
letterpaper,
top=1in,
bottom=1in,
left=1in,
right=1in]{geometry}

\newlist{enumerate*}{enumerate*}{1}
\setlist[enumerate*]{label=(\arabic*),
  after=.,
  itemjoin={{, }}, itemjoin*={{, and }}}

\newtheorem{theorem}{Theorem}[section]

\newtheorem{lemma}[theorem]{Lemma}

\newcommand{\allOne}{\mathbb{1}}

\DeclareMathOperator{\fc}{fc_{LP}}

\DeclareMathOperator{\fcSDP}{fc_{SDP}}

\newcommand{\defeq}{\coloneqq}
\theoremstyle{definition}
\newtheorem{definition}[theorem]{Definition}

\newtheorem{example}[theorem]{Example}

\theoremstyle{remark}

\newtheorem{remark}[theorem]{Remark}

\newcommand{\OPT}[2][]{\operatorname*{OPT}\sb{#1}\left(#2\right)}

\newcommand{\R}{\mathbb{R}}
\newcommand{\N}{\mathbb{N}}

\newcommand{\sprod}[2]{\langle #1, #2 \rangle}

\DeclareMathOperator{\rank}{rk}
\newcommand{\nnegrk}{\rank_{+}} % nonnegative rank
\newcommand{\LPrk}{\rank_{\textnormal{LP}}} % LP rank
\newcommand{\psdrk}{\rank_{\textnormal{psd}}} % psd rank
\newcommand{\SDPrk}{\rank_{\textnormal{SDP}}} % SDP rank
\DeclareMathOperator{\Tr}{Tr}
\DeclareMathOperator{\SOS}{SoS}
\DeclareMathOperator{\tw}{tw}
\DeclareMathOperator{\val}{val}% #(value)
\DeclareMathOperator{\corr}{corr}
\newcommand*{\size}[1]{\left|#1\right|}
\DeclareMathOperator{\argmax}{arg\,max}

\DeclareMathOperator{\dom}{dom}

\newcommand{\Problem}[1]{\textnormal{\textsf{#1}}}

\NewDocumentCommand{\MaxCSP}{mg}{%
  \ensuremath{%
    \IfValueTF{#2}
    {\Problem{Max-\(#1\)-CSP}(#2)}
    {\Problem{Max-\(#1\)-CSP}}%
    }%
  \xspace}

\DeclareExpandableDocumentCommand{\MaxXOR}{om}
{\Problem{Max-\(#2\)-XOR%
    \IfValueT{#1}{\(/#1\)}}\xspace}

\newcommand{\OneFreeCSP}{\Problem{1F-CSP}\xspace}
\newcommand{\NotEqualCSP}[1]{\Problem{\(#1\)-\(\neq\)-CSP}\xspace}

\NewDocumentCommand{\MaxCUT}{og}
{\ensuremath{\Problem{MaxCut}%
    \IfValueT{#1}{\sb{#1}}%
    \IfValueT{#2}{(#2)}}\xspace}

\NewDocumentCommand{\Matching}{og}
{\ensuremath{\Problem{Matching}%
    \IfValueT{#1}{\sb{#1}}%
    \IfValueT{#2}{(#2)}}\xspace}

\NewDocumentCommand{\UniqueGames}{O{}gg}{%
\ensuremath{%
	\IfNoValueTF{#2}
	{\Problem{UniqueGames}_{#1}}%
	{\Problem{UniqueGames}_{#1}(#2, #3)}%
	}%
	\xspace}

\NewDocumentCommand{\SparsestCut}{gg}{%
\ensuremath{%
	\IfNoValueTF{#2}{%
		\IfNoValueTF{#1}%
		{\Problem{SparsestCut}}%
		{\Problem{SparsestCut}(#1)}%
		}%
    	{\Problem{SparsestCut}(#1, #2)}%
    }%
  \xspace}

\newcommand{\BalancedSeparator}[2]%
{\Problem{BalancedSeparator(\(#1\), \(#2\))}}

\newcommand{\cube}[1]{\ensuremath{\{-1, 1\}^{#1}}}

\newcommand{\UniformVertexCover}[2]%
{\Problem{\(#1\)-UniformVertexCover(\(#2\))}}

\NewDocumentCommand{\VertexCover}{og}{%
  \ensuremath{%
    \IfValueTF{#2}
    {\Problem{VertexCover(\(#2\))}}
    {\Problem{VertexCover}}%
    \IfValueT{#1}{\sb{#1}}}%
  \xspace}
\NewDocumentCommand{\hyperVertexCover}{mog}{%
  \ensuremath{%
    \IfValueTF{#3}
    {\Problem{\(#1\)-VertexCover(\(#3\))}}
    {\Problem{\(#1\)-VertexCover}}%
    \IfValueT{#2}{\sb{#2}}}%
  \xspace}
\NewDocumentCommand{\IndependentSet}{og}{%
  \ensuremath{%
    \IfValueTF{#2}
    {\Problem{IndependentSet(\(#2\))}}
    {\Problem{IndependentSet}}%
    \IfValueT{#1}{\sb{#1}}}%
  \xspace}
\newcommand*{\set}[2]{\left\{#1\,\middle|\,#2\right\}}

\NewDocumentCommand{\conv}{mo}{\operatorname{conv}\left(#1%
    \IfValueT{#2}{\,\middle|\,#2}%
  \right)}
\NewDocumentCommand{\cone}{mo}{\operatorname{cone}\left(#1%
    \IfValueT{#2}{\,\middle|\,#2}%
  \right)}

\providecommand{\bigland}{\bigwedge}

\newcommand{\SDP}{\mathbb{S}_{+}}
\newcommand{\symM}{\ensuremath{\mathbb{S}}}

\newcommand*{\VAR}[1]{\mathbf{#1}}%variable
\NewDocumentCommand{\probability}{d()om}{%
  \operatorname{\mathbb{P}}%
  \IfValueT{#1}{\sb{#1}}%
  \left[#3%
    \IfValueT{#2}{\,\middle|\,#2}\right]}
\NewDocumentCommand{\expectation}{d()om}%
  {\operatorname{\mathbb{E}}%
      \IfValueT{#1}{\sb{#1}}%
      \left[#3%
    \IfValueT{#2}{\,\middle|\,#2}\right]}
\DeclareMathOperator{\pseudoexpectation}{\widetilde{\mathbb{E}}}

\title{Strong reductions for extended formulations}
\date{December 15, 2015}

\author[1]{Gábor Braun}
\affil[1]{ISyE, Georgia Institute of Technology,
  Atlanta, GA,
  USA.
  \textit{Email:}~gabor.braun@isye.gatech.edu}

\author[2]{Sebastian Pokutta}
\affil[2]{ISyE, Georgia Institute of Technology,
  Atlanta, GA,
  USA.
  \textit{Email:}~sebastian.pokutta@isye.gatech.edu}

\author[3]{Aurko Roy}
\affil[3]{College of Computing, Georgia Institute of Technology,
  Atlanta, GA,
  USA.
  \textit{Email:}~aurko@gatech.edu}

\begin{document}

\maketitle

\begin{abstract}
  We generalize the reduction mechanism for linear programming
  problems and semidefinite programming problems from \cite{BPZ2015} in two ways
  \begin{enumerate*}
  \item
    relaxing the requirement of affineness
  \item
    extending to fractional optimization problems
  \end{enumerate*}

  As applications we provide several new LP-hardness and SDP-hardness
  results, e.g., for the \SparsestCut problem,
  the \Problem{BalancedSeparator} problem, the \MaxCUT problem 
  and the \Matching problem on \(3\)-regular graphs.
  We also provide a new, very strong Lasserre integrality gap
  for the \IndependentSet problem, which is strictly greater than the
    best known LP approximation, showing that the Lasserre hierarchy
  does not always provide the tightest SDP relaxation.
\end{abstract}

\section{Introduction}
\label{sec:introduction}

Linear and semidefinite programs are the main components in the design
of many practical algorithms and therefore
understanding their expressive power is a fundamental problem.
The complexity of these programs is measured by the number of constraints,
ignoring all other aspects affecting the running time of an actual
algorithm, in particular, these measures are independent of the
P~vs.~NP question.
We call a problem
\emph{LP-hard} if it does not admit
an LP formulation with a polynomial number of constraints,
and we define \emph{SDP-hardness} similarly.

Recently,
motivated by Yannakakis's influential work
\cite{Yannakakis88,Yannakakis91},
a plethora of strong lower bounds have been established
for many important optimization
problems, such as e.g.,
the \Matching problem \cite{rothvoss2013matching}
or the \Problem{TravelingSalesman} problem
\cite{extform4,extform4Jour,lee2014lower}.
In \cite{BPZ2015}, the authors introduced a reduction mechanism 
providing
inapproximability results for large classes of problems.
However, the reductions were required to be affine,
and hence failed for e.g., the \VertexCover problem,
where intermediate Sherali–Adams reductions were employed
in \cite{bazzi2015no}
due to this shortcoming.

In
this work we extend the reduction mechanism of \cite{BPZ2015} in two
ways, establishing several new hardness results both
in the LP and SDP setting; both are special cases arising from reinterpreting LPs 
and
SDPs as proof systems (see
Section~\ref{sec:nonnegative-problem}). 
First, by including additional \lq{}computation\rq{} in the reduction,
we allow non-affine relations between problems,
eliminating the need for Sherali–Adams reductions in
\cite{bazzi2015no}.
Second, we extend the framework to fractional optimization
problems (such as e.g., \SparsestCut) where ratios of linear
functions have to be optimized.
Here typically one optimizes the numerator and denominator
at the same time,
and that is what we incorporate in our framework.

\subsection*{Related Work}
\label{sec:related-work}

The immediate precursor to this work is
\cite{BPZ2015} (generalizing
\cite{Pashkovich12,bfps2012}), introducing a reduction mechanism.
Base hard problems are the \Matching
problem \cite{rothvoss2013matching},
as well as constraint satisfaction problems
\cite{chan2013approximate,lee2014lower}
based on hierarchy hardness results, such as e.g.,
\cite{schoenebeck2008linear} and \cite{charikar2009integrality}.

\subsection*{Contribution}
\label{sec:contribution}

\begin{description}[leftmargin=1em, font=\normalfont\emph]
\item[Generalized LP/SDP reductions.]
  We generalize the reduction
  mechanism in \cite{BPZ2015} by modeling additional computation,
  i.e., using extra LP or SDP constraints. Put differently, we allow
  for more complicated reduction maps as long as these maps themselves
  have a small
  LP/SDP formulation.  As a consequence, we can
 relax the affineness requirement and enable
  a weak form of \emph{gap-amplification} and
  \emph{boosting}. This overcomes a major limitation of the approach in
  \cite{BPZ2015}, yielding significantly stronger reductions
  at a small cost.

\item[Fractional LP/SDP optimization.]
  Second, we present a
  \emph{fractional} LP/SDP framework and reduction mechanism,
  where the objective functions are ratios of functions from
  a low dimensional space, such as for the \SparsestCut problem.
  For these problems
  the standard LP/SDP framework is meaningless as the ratios span a
  high dimensional affine space.
  The fractional framework models the
  usual way of solving fractional optimization problems,
  enabling us to establish strong
  statements about LP or SDP complexity.

\item[Direct non-linear hardness reductions.]  We
  demonstrate the power of our generalized reduction by establishing
  new LP-hardness and SDP-hardness for several problems of interest,
  i.e., these problems cannot be solved by LPs/SDPs of polynomial size;
  see Table~\ref{fig:reductions}.  We establish various
  hardness results for the \SparsestCut and
  \Problem{BalancedSeparator}
  problems
  even when one of the underlying graph has bounded treewidth.  We
  redo the reductions to intermediate CSP problems used for optimal
  inapproximability results for the \VertexCover problem over simple
  graphs and \(Q\)-regular
  hypergraphs in \cite{bazzi2015no}, eliminating Sherali–Adams
  reductions.  We also show the first explicit SDP-hardness for the
  \MaxCUT problem, inapproximability within a factor of
  \(15/16+\varepsilon\),
  which is stronger than the algorithmic hardness of
  \(16/17+\varepsilon\).
  Finally, we prove a new, strong Lasserre integrality gap of
  \(n^{1-\gamma}\)
  after \(O(n^\gamma)\)
  rounds for the \IndependentSet problem for any sufficiently small
  \(\gamma > 0\).
  It not only significantly strengthens and complements the best-known
  integrality gap results so far
  (\cite{tulsiani2009csp} and
  \cite{au2011complexity,au2013comprehensive}; see also
  \cite{liptak2003stable,stephen1999representation}),
  but also shows the
  suboptimality of Lasserre relaxations for the \IndependentSet
  problem together with \cite{bazzi2015no}.

\item[Small uniform LPs for bounded treewidth problems.]
Finally, we introduce a new technique in
  Section~\ref{sec:small-uniform-lp} to derive small
  \emph{uniform} linear programs for problems over graphs of bounded
  treewidth.  Here the same linear program is used for all bounded
  treewidth instances of the same size, independent of the actual tree
  decompositions, whereas the linear program in \cite{Kolman15}
  work for a single input instance only (with fewer inequalities
  than our linear program).
\end{description}

\begin{table*}[h]
\renewcommand{\arraystretch}{1.2}
  \centering \footnotesize
  \begin{tabulary}{\textwidth}{LclcC}
    Problem & Factor & Source & Paradigm & Remark
    \\ \hline 
    \MaxCUT & \(\frac{15}{16} + \varepsilon\) & \MaxXOR{3}/0 & SDP
	\\
    \SparsestCut{n},
    \(\tw(\text{supply}) = O(1)\)
    &
    \(2 - \varepsilon\) & \MaxCUT & LP & opt.   
                                         \cite{chan2013approximate}
  	\\    
    \SparsestCut{n}, \(\tw(\text{supply}) = O(1)\)
    &
    \(\frac{16}{15} - \varepsilon\) & \MaxCUT & SDP 
	% \\
    \\
    \BalancedSeparator{n}{d}, \(\tw(\text{demand}) = O(1)\) &
    \(\omega(1)\) & \UniqueGames & LP
    \\
    \IndependentSet & \(\omega(n^{1-\varepsilon})\) & \MaxCSP{k} &
    \begin{tabular}{c}
      Lasserre \\ \(O(n^\varepsilon)\) rounds
    \end{tabular}
                                                            & 
    \\
    \Matching, 3-regular & \(1 + \varepsilon / n^{2}\) & \Matching
    & LP
    \\
    \OneFreeCSP \tikzmark{1F-CSP} & \multirow{2}{*}{\(\omega(1)\)} &
    \multirow{2}{*}{\UniqueGames} & \multirow{2}{*}{LP} &
    \multirow{2}{*}{%
      \begin{tabular}{l}
        \cite{bazzi2015no} \\ w/o SA
      \end{tabular}
    }
    \\
    \NotEqualCSP{Q} \tikzmark{NotEqCSP}&
  \end{tabulary}
  %% Add brace
  % Source: http://tex.stackexchange.com/questions/105934/brace-spanning-multiple-rows-in-a-table-and-horizontal-alignment
  \begin{tikzpicture}[overlay, remember picture]
    \draw[decorate,decoration={brace,amplitude=1ex}]
    ($(1F-CSP -| NotEqCSP) + (0.0em,+1.7ex)$ ) to
    ($(NotEqCSP) + (0.0em,-0.5ex)$);
  \end{tikzpicture}
  \vspace{8pt}
  \caption{\label{fig:reductions} Inapproximability of optimization
    problems.
    \(\tw\) denotes treewidth.}
\end{table*}
\subsection*{Outline}
\label{sec:outline}

We start by recalling and refining the linear programming framework
in Section~\ref{sec:preliminaries}, including the
optimization problems we shall consider.
We develop a general theory in Section~\ref{sec:nonnegative-problem}
leading easily to both a generalized reduction mechanism in
Section~\ref{sec:reductions} and
an extension to fractional optimization
in Section~\ref{sec:fract-opt-problems-1}.
The remaining chapters contain mostly applications
to various problems.
Exceptions are Section~\ref{sec:Lasserre-independent-set},
establishing a Lasserre integrality gap
for the \IndependentSet problem,
and Section~\ref{sec:small-uniform-lp}
providing a small linear program for bounded treewidth problems.

\section{Preliminaries}
\label{sec:preliminaries}

Here we recall the linear programming
and semidefinte programming framework from \cite{BPZ2015},
as well as the optimization problems we shall consider later,
paying particular attention to base hard problems.
Section~\ref{sec:nonnegative-problem} is a new technical foundation 
for
the framework, presenting the underlying theory in a unified simple
way, from which the extensions in Sections~\ref{sec:reductions} and
\ref{sec:fract-opt-problems-1} readily follow. We start by
recalling the notion of \emph{tree decompositions} and 
\emph{treewidth} of a graph.

\begin{definition}[Tree width]
\label{def:treewidth}
A \emph{tree decomposition} of a graph \(G\)
is a tree \(T\) together with a vertex set of \(G\) called
\emph{bag} \(B_{t} \subseteq V(G)\)
for every node \(t\) of \(T\),
satisfying the following conditions:
\begin{enumerate*}
\item \(V(G) = \bigcup_{t \in V(T)} B_{t}\)
\item For every adjacent vertices \(u\), \(v\) of \(G\)
  there is a bag \(B_{t}\) containing both \(u\) and \(v\)
\item For all nodes \(t_{1}\), \(t_{2}\), \(t\) of \(T\)
  with \(t\) lying between \(t_{1}\) and \(t_{2}\)
  (i.e., \(t\) is on the unique path connecting
  \(t_{1}\) and \(t_{2}\))
  we have \(B_{t_{1}} \cap B_{t_{2}} \subseteq B_{t}\)
\end{enumerate*}
The \emph{width} of the tree decomposition is
\(\max_{t \in V(T)} \size{B_{t}} - 1\):
one less than the maximum bag size.
The \emph{treewidth} \(\tw(G)\) of \(G\)
is the minimum width of its tree decompositions.
\end{definition}

We will use
\(\chi(\cdot)\) for indicator functions:
i.e., \(\chi(X) = 1\) if the statement \(X\) is true,
and \(\chi(X) = 0\) otherwise.
We will denote random
variables using bold face, e.g. \(\VAR{x}\).
Let \(\symM^{r}\) denote the set of symmetric \(r \times r\) real
matrices, and let \(\SDP^{r}\) denote the set of positive semidefinite
\(r \times r\) real matrices.

\subsection{Optimization Problems}
\label{sec:optim-probl}

\begin{definition}[Optimization problem]
  \label{def:opt-problem}
  An \emph{optimization problem}
  is a tuple
  \(\mathcal{P} = (\mathcal{S}, \mathfrak{I}, \val)\)
  consisting of a set \(\mathcal{S}\)
  of \emph{feasible solutions}, a set \(\mathfrak{I}\)
  of \emph{instances},
  and a real-valued objective
  called \emph{measure}
  \(\val \colon \mathfrak{I} \times \mathcal{S}  \to \R\).

We shall write \(\val_{\mathcal{I}}(s)\) for the objective value
of a feasible solution \(s \in \mathcal{S}\)
for an instance \(\mathcal{I} \in \mathfrak{I}\).
\end{definition}

The \SparsestCut problem is defined over a graph with two kinds of
edges: \emph{supply} and \emph{demand} edges. The objective is to 
find a cut that minimizes
the ratio of the capacity of cut supply edges 
to the total demand separated.
For a weight function \(f \colon E(K_n) \rightarrow
\R_{\ge 0}\), we define the graph \([n]_f \defeq ([n], E_f)\)
where \(E_f \defeq \{(i, j) \mid i, j \in [n], f(i, j) > 0\}\).
We study the \SparsestCut problem with 
bounded-treewidth supply graph.

\begin{definition}[\SparsestCut{n}{k}]
\label{def:sparsestCut}
  Let \(n\) be a positive integer. The minimization problem \SparsestCut{n}{k} consists of
  \begin{description}
  \item[instances]
    a pair \((d, c)\) of a
    nonnegative \emph{demand} \(d \colon E(K_n)\rightarrow
    \R_{\ge 0}\) 
    and a \emph{capacity}
    \(c \colon E(K_{n}) \to \R_{\ge 0}\)
    such that \(\tw([n]_c) \le k\);
  \item[feasible solutions]
    all subsets \(s\) of \([n]\);
  \item[measure]
    ratio of separated capacity and separated demand:
    \begin{equation*}
      \val_{d, c}(s) = \frac{\sum_{i \in s, j \notin s} c(i, j)}
      {\sum_{i \in s, j \notin s} d(i, j)}
     \end{equation*}
	for capacity \(c\), demand \(d\), and set \(s\). 
  \end{description}
\end{definition}

The \Problem{BalancedSeparator} problem 
is similar to the \SparsestCut
problem and is also defined over
a graph with \emph{supply} and \emph{demand} edges.
However it restricts the solutions to cuts that are \emph{balanced},
i.e., which separate a large proportion of the demand. Note
that in this case we define the \Problem{BalancedSeparator}
 problem on \(n\) vertices for a 
 fixed demand function \(d\), unlike in the case
of \SparsestCut where the demand function \(d\) was part of the
instances. This is because in the framework of \cite{BPZ2015}
the solutions should be independent of the instances. We formalize
this below. 

\begin{definition}[\BalancedSeparator{n}{d}]
  \label{def:balanced-separator}
  Let \(n\) be a positive integer, and
  \(d \colon [n] \times [n] \to \R_{\ge 0}\) a nonnegative function
  called \emph{demand function}.
  Let \(D\) denote the total demand \(\sum_{i, j \in [n]} d(i, j)\).
  The minimization problem \BalancedSeparator{n}{d} consists of
  \begin{description}
  \item[instances]
    nonnegative \emph{capacity} function
    \(c \colon E(K_{n}) \to \R_{\ge 0}\)
    on the edges of the complete graph \(K_{n}\);
  \item[feasible solutions]
    all subsets \(s\) of \([n]\)
    such that \(\sum_{i \in s, j \notin s} d(i, j)\)
     is at least \(D/4\);
  \item[measure]
    capacity of cut supply edges:
    \(\val_{c}(s) \defeq
     \sum_{i \in s, j \notin s} c(i, j)\)
    for a capacity function \(c\) and set \(s\).
  \end{description}
\end{definition}

Recall that an \emph{independent set} \(I\) of a graph \(G\)
is a subset of pairwise non-adjacent vertices
\(I \subseteq V(G)\). The \IndependentSet problem
on a graph \(G\) asks for an independent set of \(G\) of maximum size.
We formally define it as an optimization problem below.

\begin{definition}[\IndependentSet{G}]
  \label{def:independent-set}
  Given a graph \(G\), the maximization problem \IndependentSet{G}
  consists of
  \begin{description}
  \item[instances] all induced subgraphs \(H\) of \(G\);
  \item[feasible solutions] all independent subsets \(I\) of \(G\);
  \item[measure] \(\val_{H}(I) = \size{I \cap V(H)}\).
  \end{description}
\end{definition}

Recall that a subset \(X\) of \(V(G)\) for a graph \(G\) is a 
\emph{vertex cover} if every edge of \(G\) has at least
one end point in \(X\). The \VertexCover problem on a graph \(G\)
asks for a vertex cover of \(G\) of minimum size. We give a formal
definition below.

\begin{definition}[\VertexCover{G}]
\label{def:vertex-cover}
Given a graph \(G\), the minimization problem \VertexCover{G}
consists of
\begin{description}
\item[instances] all induced subgraphs \(H\) of \(G\);
\item[feasible solutions] all vertex covers \(X\) of \(G\);
\item[measure] \(\val_{H}(X) = \size{X \cap V(H)}\).
\end{description}
\end{definition}

The \MaxCUT problem on a graph \(G\) asks for a vertex set of \(G\)
cutting a maximum number of edges.
Given a vertex set \(X \subseteq V(G)\),
let \(\delta_G(X) \defeq
\set{\{u, v\} \in E(G)}{u \in X, v \notin X}\)
denote the set of edges of \(G\) with one end point in \(X\)
and the other end point outside \(X\).
\begin{definition}[\MaxCUT{G}]
  \label{def:maxcut}
  Given a graph \(G\), the 
  maximization problem \MaxCUT{G}
  consists of
  \begin{description}
  \item[instances] all induced subgraph \(H\) of \(G\);
  \item[feasible solutions] all vertex subsets \(X \subseteq V(G)\);
  \item[measure] \(\val_{H}(X) = \size{E(H) \cap \delta_G(X)}\).
  \end{description}
\end{definition}

\emph{Constraint satisfaction problems} (CSPs for short)
are inherently related to inapproximability results,
and form a basic collection of inapproximable problems.
There are many variants of CSPs,
but the general structure is as follows:
\begin{definition}[Constraint Satisfaction Problems]
  A \emph{constraint satisfaction problem}, in short CSP,
  is an optimization problem on a fixed set
  \(\{x_{1}, \dotsc, x_{n}\}\)
  of variables with values in a fixed set \([q]\)
  consisting of
  \begin{description}
  \item[instances]
    formal weighted sums
    \(\mathcal{I} = \sum_{i} w_{i} C_{i}(x_{j_{1}}, 
    \dotsc, x_{j_{k_{i}}})\)
    of some clauses \(C_{i} \colon [q]^{k_{i}} \to \{0, 1\}\)
    with weights \(w_{i} \geq 0\).
  \item[feasible solutions]
    all mappings \(s \colon \{x_{1}, \dotsc, x_{n}\} \to [q]\),
    called \emph{assignments} to variables
  \item[measure]
    weighted fraction of satisfied clauses:
    \begin{equation*}
      \val_{\mathcal{I}}(s) \coloneqq \frac{
        \sum_{i} w_{i} C_{i}(s(x_{j_{1}}), \dotsc, s(x_{j_{k_{i}}}))
      }{\sum_{i} w_{i}}
    \end{equation*}
  \end{description}
\end{definition}
A CSP can be either a maximization problem or a minimization problem.
For specific CSPs there are restrictions on permitted clauses,
and later we will define CSPs by specifying only these restrictions.
For example \MaxCSP{k} is the problem where only clauses with
at most \(k\) free variables are allowed (i.e., \(k_{i} \leq k\)
in the definition above).
The problem \MaxXOR{k} is the problem with clauses of the form
\(x_{1} + \dotsb + x_{k} = b\)
where the \(x_{i}\) are distinct variables, \(b \in \{0, 1\}\),
and the addition is modulo \(2\).
We shall use the subproblem \MaxXOR[0]{k},
where the clauses have the form
\(x_{1} + \dotsb + x_{k} = 0\).

Given a \(k\)-ary predicate \(P\),
let \(\MaxCSP{k}{P}\) denote the CSP where all clauses arise
via a change of variables from \(P\),
i.e., every clause have the form \(P(x_{i_{1}}, \dotsc, x_{i_{k}})\)
with \(i_{1}\), \dots, \(i_{k}\) being pairwisely distinct.
For example,
\(\MaxXOR[0]{k} = \MaxCSP{k}{x_{1} + \dotsb + x_{k} = 0}\).

Another specific example of a CSP we will make use of is the
\UniqueGames problem. 
The \UniqueGames problem asks for a labeling of the vertices of a graph
that maximizes the number (or weighted sum) of edges where the labels of
the endpoints match. We formalize it
restricted to regular bipartite graphs.

\begin{definition}[{\UniqueGames[\Delta]{n}{q}}]
  \label{def:Unique-Games}
  Let \(n\), \(q\) and \(\Delta\) be positive integer parameters.
  The maximization problem \UniqueGames[\Delta]{n}{q}
  consists of
  \begin{description}
  \item[instances]
    All edge-weighted \(\Delta\)-regular bipartite graphs \((G, w)\)
    (i.e., a graph \(G\) with a collection
    \(\{w_{u,v}\}_{\{u, v\} \in E(G)}\) of real numbers)
    with partite sets
    \(\{0\} \times [n]\) and \(\{1\} \times [n]\)
    with every edge \(\{i, j\}\) labeled with a permutation
    \(\pi_{i,j} \colon [q] \to [q]\)
    such that \(\pi_{i,j} = \pi_{j,i}^{-1}\).
  \item[feasible solutions]
    All functions \(s \colon \{0,1\} \times [n] \to [q]\)
    called \emph{labelings} of the vertices.
  \item[measure]
    The weighted fraction of correctly labeled edges,
    i.e., edges \(\{i, j\}\) with \(s(i) = \pi_{i, j}(s(j))\):
    \begin{equation*}
      \val_{(G, w)}(s) \coloneqq
      \frac{%
        \sum_{\substack{\{i,j\} \in E(G) \\ s(i) = \pi_{i, j}(s(j))}}
      w(i,j)}
    {\sum_{\{i,j\} \in E(G)}w(i,j)}
    \end{equation*}
  \end{description}
\end{definition}

The \Matching problem asks for a matching in a graph \(H\)
of maximal size.  The restriction to matchings and subgraphs (which
corresponds to 0/1 weights in the objective of the matching problem) below serves the purpose to obtain a
base hard problem, with which we can work more easily later.

\begin{definition}[\Matching{G}]\label{prob:pm}  
  The maximum
  matching problem \(\Matching{G}\) over
  a graph \(G\) is defined as
  the maximization problem:
\begin{description}
\item[instances]
  all subgraphs \(H\) of \(G\)
\item[feasible solutions] all perfect matchings
  \(S\) on \(G\).
\item[measure]
  the size of induced matching
  \(\val_{G}(S) \coloneqq \size{S \cap E(H)}\) with
  \(S \in \mathcal{S}\), and
  \(H\) a subgraph of \(G\).
\end{description}
We will also write \Matching[k]{G} to indicate that the maximum vertex
degree is at most \(k\). 
\end{definition}

\subsubsection{Uniform problems}
\label{sec:uniform-problems}

Here we present so called \emph{uniform} versions of some of the
optimization problems discussed so far, where the class of instances
is typically much larger, e.g., the class of all instances of a given
size.  \emph{Non-uniform} optimization problems typically consider
weighted versions of a \emph{specific instance} or all induced
subgraphs of a \emph{given graph}. For establishing lower bounds,
non-uniform optimization problems give stronger bounds: \lq{}even if
we consider a \emph{specific graph}, then there is no small
LP/SDP\rq{}. In the case of upper bounds, i.e., when we provide
formulations, uniform optimization problems provide stronger
statements: \lq{}even if we consider \emph{all graphs
  simultaneously}, then there exists a small LP/SDP\rq{}.

We will
later show in Section~\ref{sec:small-uniform-lp} that over graphs of
bounded tree-width there exists a small LP that solves the
uniform version of optimization problems.  We start by defining the uniform version of \MaxCUT.
Recall that for a graph \(G\) and a subset \(X\) of \(V(G)\),
we define \(\delta_G(X) \defeq \set{\{u, v\} \in E(G)}{
u \in X, v \notin X}\) to be the set of crossing edges.

\begin{definition}[\MaxCUT{n}]
\label{def:max-cut-uniform}
For a positive integer \(n\), the maximization problem \MaxCUT{n}
consists of
\begin{description}
\item[instances] all graphs \(G\) with \(V(G) \subseteq [n]\);
\item[feasible solutions] all subsets \(X\) of \([n]\);
\item[measure] \(\val_{G}(X) = \size{\delta_G(X)}\).
\end{description}
\end{definition}

With \IndependentSet and \VertexCover
we face the difficulty that the solutions are instance dependent.
Hence we enlarge the feasible solutions to include all possible 
vertex sets, and in the objective function penalize the violation
of requirements.
\begin{definition}[\IndependentSet{n}]
  \label{def:independent-set-uniform}
  For a positive integer \(n\),
  the maximization problem \IndependentSet{n}
  consists of
  \begin{description}
  \item[instances] all graphs \(G\) with \(V(G) \subseteq [n]\);
  \item[feasible solutions] all subsets \(X\) of \([n]\);
  \item[measure]
    the number of vertices of \(G\) in \(X\)
    penalized by the number of edges of \(G\)
    inside \(X\):
    \begin{equation}
      \val_{G}(X) = \size{X \cap V(G)} - \size{E(G[X])}.
    \end{equation}
  \end{description}
\end{definition}

Recall that \VertexCover asks for a minimal size vertex set \(X\)
of a graph \(G\) such that every edge of \(G\) has at least one of its
endpoints in \(X\).
\begin{definition}
  \label{def:vertex-cover}
  For a positive integer \(n\)
  the minimization problem \VertexCover consists of
  \begin{description}
  \item[instances] all graphs \(G\) with \(V(G) \subseteq [n]\)
  \item[feasible solutions]
    all subsets \(X \subseteq V(G)\)
  \item[measure]
    the number of vertices of \(G\) in \(X\)
    penalized by the number of uncovered edges:
    \begin{equation}
      \val_{G}(X) \coloneqq
      \size{X \cap V(G)}
      + \size{E(G \setminus X)}
      .
    \end{equation}
  \end{description}
\end{definition}

\subsection{Nonnegativity problems: Extended formulations as proof system}
\label{sec:nonnegative-problem}

In this section we introduce an abstract view of formulation
complexity, where the main idea is to reduce all statements to 
the core question about the complexity of deriving nonnegativity
for a class of nonnegative functions.
This abstract view will allow us to easily  introduce
future versions of reductions and optimization problems with automatic availability of
Yannakakis's Factorization Theorem and the reduction mechanism.

\begin{definition}
  \label{def:nonnegative-problem}
  A \emph{nonnegativity problem}
  \(\mathcal{P} = (\mathcal{S}, \mathfrak{I}, \val)\)
  consists of
  a set \(\mathfrak{I}\) of \emph{instances},
  a set \(\mathcal{S}\) of \emph{feasible solutions}
  and a nonnegative \emph{evaluation}
  \(\val \colon \mathfrak{I} \times \mathcal{S} \to \R_{\geq 0}\).
\end{definition}

As before, we shall write \(\val_{\mathcal{I}}(s)\) instead of
\(\val(\mathcal{I}, s)\). The aim is to study the
complexity of proving nonnegativity of
the functions \(\val_{\mathcal{I}}\).
Therefore we define the notion of proof as a linear program or a
semidefinite program.

\begin{definition}
  \label{def:nonneg-LP-proof}
  Let \(\mathcal{P} = (\mathcal{S}, \mathfrak{I}, \val)\)
  be a nonnegativity problem.
  An \emph{LP proof} of nonnegativity of \(\mathcal{P}\)
  consists of a linear program \(A x \leq b\)
  with \(x \in \R^{r}\) for some \(r\)
  and the following \emph{realizations}:
  \begin{description}
  \item[Feasible solutions] as vectors \(x^{s} \in \R^{r}\)
    for every \(s \in \mathcal{S}\) satisfying
  \begin{align}
    \label{eq:nonneg-LP-contain}
    A x^{s} &\leq b \qquad \text{for all } s \in \mathcal{S},
  \end{align}
  i.e., the system \(Ax \leq b\) is a relaxation (superset) of
  \(\conv{x^s \mid s \in \mathcal{S}}\).
  \item[Instances] as affine functions
    \(w_{\mathcal{I}} \colon \R^{r} \to \R\)
    for all \(\mathcal{I} \in \mathfrak{I}^{S}\)
    satisfying
    \begin{equation}
      \label{eq:nonneg-LP-linear}
      w_{\mathcal{I}}(x^{s}) = \val_{\mathcal{I}}(s)
      \qquad \text{for all } s \in
      \mathcal{S},
    \end{equation}
    i.e., the linearization \(w_{\mathcal{I}}\) of
    \(\val_{\mathcal{I}}\)
    is required to be exact on all \(x^s\)
    with \(s \in \mathcal{S}\).
  \item[Proof] 
    We require that the \(w_{\mathcal{I}}\) are nonnegative
    on the solution set of the LP:
    \begin{equation}
      \label{eq:nonneg-LP-nonneg}
      w_{\mathcal{I}}(x) \geq 0 \qquad \text{whenever } A x \leq b,
      \mathcal{I} \in \mathfrak{I}
      .
    \end{equation}
  \end{description}
  The \emph{size} of the formulation is the number of inequalities
  in \(A x \leq b\).
  Finally,
  \emph{LP proof complexity} \(\fc(\mathcal{P})\)
  of \(\mathcal{P}\) is
  the minimal size of all its LP proofs.
\end{definition}

The notion of an SDP proof is defined similarly.

\begin{definition}
  \label{def:nonneg-SDP-proof}
  Let \(\mathcal{P} = (\mathcal{S}, \mathfrak{I}, \val)\)
  be a nonnegativity problem.
  An \emph{SDP proof} of nonnegativity of \(\mathcal{P}\)
  consists of a semidefinite program
  \(\set{X \in \SDP^{r}}{\mathcal{A}(X) = b}\)
  (i.e., a linear map \(\mathcal{A} \colon \symM^{r} \to \R^{k}\)
  together with a vector \(b \in \R^{k}\))
  and the following \emph{realizations}:
  \begin{description}
  \item[Feasible solutions] as vectors \(X^{s} \in \SDP^{r}\)
    for all \(s \in \mathcal{S}\) satisfying
    \begin{align}
      \label{eq:nonneg-SDP-contain}
      \mathcal{A}(X^{s}) = b
    \end{align}
  \item[Instances] as nonnegative affine functions
    \(w_{\mathcal{I}} \colon \symM^{r} \to \R\)
    for all \(\mathcal{I} \in \mathfrak{I}\)
    satisfying
    \begin{equation}
      \label{eq:nonneg-SDP-linear}
      w_{\mathcal{I}}(X^{s}) = \val_{\mathcal{I}}(s)
      \qquad \text{for all } s \in \mathcal{S}
      .
    \end{equation}
  \item[Proof] We require nonnegativity on the feasible region of the
    SDP:
    \begin{equation}
      \label{eq:nonneg-SDP-nonneg}
      w_{\mathcal{I}}(X) \geq 0 \qquad \text{whenever }
      \mathcal{A}(X) = b,
      X \in \SDP^{r},
      \mathcal{I} \in \mathfrak{I}
      .
    \end{equation}
  \end{description}
  The \emph{size} of the formulation is the dimension parameter \(r\).
  Finally, the
  \emph{SDP proof complexity} \(\fcSDP(\mathcal{P})\)
  of \(\mathcal{P}\) is
  the minimal size of all its SDP proofs.
\end{definition}

\subsubsection{Slack matrix and proof complexity}
\label{sec:slack-matrix-proof}

We introduce the slack matrix of a nonnegativity problem
as a main tool to study proof complexity,
generalizing the approach from the polyhedral world.
The main result is a version of Yannakakis's Factorization Theorem
formulating proof complexity in the language of linear algebra
as a combinatorial property of the slack matrix.
\begin{definition}
  \label{def:nonneg-slack-matrix}
  The \emph{slack matrix} of a nonnegativity problem
  \(\mathcal{P} = (\mathcal{S}, \mathfrak{I}, \val)\)
  is the \(\mathfrak{I} \times \mathcal{S}\) matrix
  \(M_{\mathcal{P}}\)
  with entries the values of the function \(\val_{\mathcal{I}}\)
  \begin{equation}
    \label{eq:nonneg-slack-matrix}
    M_{\mathcal{P}}(\mathcal{I}, s) \coloneqq \val_{\mathcal{I}}(s)
    .
  \end{equation}
\end{definition}

We will use the standard notions of nonnegative rank and semidefinite
rank. 

\begin{definition}[\cite{BPZ2015}]
  \label{def:matrix-factorization-nonneg}
  Let \(M\) be a nonnegative matrix.
  \begin{description}
  \item[nonnegative factorization]
    A \emph{nonnegative factorization} of \(M\) of size \(r\)
    is a decomposition \(M = \sum_{i=1}^{r} M_{i}\) of \(M\)
    as a sum of \(r\) nonnegative matrices \(M_{i}\) of rank \(1\).
    The \emph{nonnegative rank} \(\nnegrk M\) is the minimum \(r\)
    for which \(M\) has a nonnegative factorization of size \(r\).
  \item[psd factorization]
    A \emph{positive semi-definite (psd) factorization} of \(M\)
    of size \(r\)
    is a decomposition \(M(\mathcal{I}, s) =
    \Tr[A_{\mathcal{I}} B_{s}]\) of \(M\)
    where the \(A_{\mathcal{I}}\) and \(B_{s}\)
     are positive semi-definite (psd)
    \(r \times r\) matrices.
    The \emph{psd rank} \(\psdrk M\) is the minimum \(r\)
    for which \(M\) has a psd factorization of size \(r\).
  \end{description}
  We define variants ignoring factors of the form \(a \allOne\):
  \begin{description}
  \item[LP factorization]
    An \emph{LP factorization} of \(M\) of size \(r\)
    is a decomposition \(M = \sum_{i=1}^{r} M_{i} + u \allOne\)
    of \(M\)
    as a sum of \(r\) nonnegative matrices \(M_{i}\) of rank \(1\)
    and possibly an additional nonnegative rank-\(1\)
    \(u \allOne\) with all columns being equal.
    The \emph{LP rank} \(\LPrk M\) is the minimum \(r\)
    for which \(M\) has an LP factorization of size \(r\).
  \item[SDP factorization]
    An \emph{SDP factorization} of \(M\) of size \(r\)
    is a decomposition \(M(\mathcal{I},s) =
    \Tr[A_{\mathcal{I}} B_{s}] + u_{\mathcal{I}} \)
    of \(M\)
    where the \(A_{\mathcal{I}}\) and \(B_{s}\)
     are positive semi-definite (psd)
    \(r \times r\) matrices,
    and \(u_{\mathcal{I}}\) is a nonnegative number.
    The \emph{SDP rank} \(\SDPrk M\) is the minimum \(r\)
    for which \(M\) has an SDP factorization of size \(r\).
  \end{description}
\end{definition}

\begin{remark}
  \label{rem:matrix-factorization}
  The difference between LP rank and nonnegative rank (see Definition~\ref{def:matrix-factorization-nonneg}) is
  solely by measuring the size of a factorization:
  for LP rank factors with equal columns
  do not contribute to the size.
  This causes a difference of at most \(1\) between the two ranks.
  The motivation for the LP rank is that it captures exactly the
  LP formulation complexity of an optimization problem, in particular
  for approximation problems (see \cite{BPZ2015} for an in-depth discussion).
  Similar remarks apply to the relation of SDP rank, psd rank,
  and SDP formulation complexity.
\end{remark}

\begin{theorem}
  \label{thm:factorization-nonneg}
  For every nonnegativity problem \(\mathcal{P}\)
  with slack matrix \(M_{\mathcal{P}}\) we have
  \begin{align}
    \label{eq:factorization-nonneg_LP}
    \fc(\mathcal{P})
    &
    = \LPrk M_{\mathcal{P}}
    ,
    \\
    \label{eq:factorization-nonneg_SDP}
    \fcSDP(\mathcal{P})
    &
    = \SDPrk M_{\mathcal{P}}
    ,
  \end{align}
\begin{proof}
The proof is an extension of the usual proofs of Yannakakis's
Factorization Theorem, e.g., that in \cite{BPZ2015}.
We provide the proof only for the LP case,
as the proof for the SDP case is similar.

First we prove \(\LPrk M_{\mathcal{P}} \leq \fc(\mathcal{P})\).
Let  \(Ax \leq b\) be an LP proof for \(\mathcal{P}\) of size
\(\fc(\mathcal{P})\)
with realization \(x^s\) for \(s \in \mathcal{S}\)
and affine functions \(w_{\mathcal{I}}\) for
\(\mathcal{I} \in \mathfrak{I}\).
By Farkas's lemma,
there are nonnegative matrices
\(u_{\mathcal{I}}\) and nonnegative numbers \(\gamma_{\mathcal{I}}\)
with
\(w_{\mathcal{I}}(x) = u_{\mathcal{I}} \cdot (b - Ax) +
\gamma_{\mathcal{I}}\).
Substituting \(x\) by \(x^{s}\), we obtain
an LP factorization of size \(\fc(\mathcal{P})\):
\begin{equation*}
  M_{\mathcal{P}} (\mathcal{I}, s) = \val_{\mathcal{I}}(s)
  = w_{\mathcal{I}}(x^{s})
  = u_{\mathcal{I}} \cdot (b - A x^{s}) + \gamma_{\mathcal{I}}
  .
\end{equation*}

Conversely, to show
\(\fc(\mathcal{P}) \leq \LPrk(\mathcal{P})\),
we choose an LP factorization
of \(M_{\mathcal{P}}\) of size \(r = \LPrk(\mathcal{P})\)
\begin{equation*}
  M_{\mathcal{P}}(\mathcal{I}, s) = u_{\mathcal{I}}
  x^{s} + \gamma_{\mathcal{I}}
\end{equation*}
where the \(u_{\mathcal{I}}\) and \(x^{s}\)
are nonnegative matrices of size \(1 \times r\) and \(r \times 1\),
respectively,
and the \(\gamma_{\mathcal{I}}\) are nonnegative numbers.
Now \(\mathcal{P}\) has the following LP proof:
The linear program is \(x \geq 0\) for \(x \in \R^{r \times 1}\).
A feasible solution \(s\) is represented by the vector \(x^{s}\).
An instance \(\mathcal{I}\) is represented by
\begin{equation*}
  w_{\mathcal{I}}(x) \defeq u_{\mathcal{I}} x + \gamma_{\mathcal{I}}
  .
\end{equation*}
To check the proof, note that by nonnegativity of \(u_{\mathcal{I}}\)
and \(\gamma_{\mathcal{I}}\), we have
\(w_{\mathcal{I}}(x) \geq 0\) for all \(x \geq 0\).
Clearly,
\(w_{\mathcal{I}}(x^{s}) = M_{\mathcal{P}}(\mathcal{I}, s)
= \val_{\mathcal{I}}(s)\),
completing the proof.
\end{proof}
\end{theorem}

\subsubsection{Reduction between nonnegativity problems}
\label{sec:reduct-nonnegative}

\begin{definition}[Reduction]
  \label{def:red-nonneg}
  Let \(\mathcal{P}_{1} = (\mathcal{S}_{1}, \mathfrak{I}_{1},
  \val^{\mathcal{P}_{1}})\)
  and
  \(\mathcal{P}_{2} = (\mathcal{S}_{2}, \mathfrak{I}_{2},
  \val^{\mathcal{P}_{2}})\)
  be nonnegativity problems.

  A \emph{reduction} from \(\mathcal{P}_{1}\) to \(\mathcal{P}_{2}\)
  consists of
  \begin{enumerate}
  \item
    two mappings:
    \(* \colon \mathfrak{I}_{1} \to \mathfrak{I}_{2}\)
    and
    \(* \colon \mathcal{S}_{1} \to \mathcal{S}_{2}\)
    translating instances and feasible solutions independently;
  \item
    two nonnegative \(\mathfrak{I}_{1} \times \mathcal{S}_{1}\)
    matrices \(M_{1}\), \(M_{2}\)
  \end{enumerate}
  satisfying
  \begin{equation}
    \label{eq:red-nonneg}
    \val^{\mathcal{P}_{1}}_{\mathcal{I}_{1}}(s_{1})
    =
    \val^{\mathcal{P}_{2}}_{\mathcal{I}_{1}^{*}}(s_{1}^{*})
    \cdot
    M_{1}(\mathcal{I}_{1}, s_{1})
    +
    M_{2}(\mathcal{I}_{1}, s_{1})
    .
  \end{equation}
\end{definition}
The matrices \(M_{1}\) and \(M_{2}\) encode additional arguments in the
nonnegativity proof of \(\mathcal{P}_{1}\),
besides using nonnegativity of \(\mathcal{P}_{2}\).
Therefore in applications they should have low complexity,
to provide a strong reduction.
The following theorem relates the proof complexity of problems
in a reduction.
\begin{theorem}
  \label{thm:red-nonneg}
  Let \(\mathcal{P}_{1}\) and \(\mathcal{P}_{2}\) be
  nonnegativity problems
  with a reduction from \(\mathcal{P}_{1}\) to \(\mathcal{P}_{2}\).
  Then
  \begin{align}
    \fc(\mathcal{P}_{1})
    &
    \leq
    \LPrk M_{2} +
    \LPrk M_{1}
    + \nnegrk M_{1} \cdot \fc(\mathcal{P}_{2}),
    \\
    \fcSDP(\mathcal{P}_{1})
    &
    \leq
    \SDPrk M_{2} +
    \SDPrk M_{1}
    + \psdrk M_{1} \cdot \fcSDP(\mathcal{P}_{2}),
  \end{align}
  where \(M_{1}\) and \(M_{2}\) are the matrices in the reduction
  as in Definition~\ref{def:red-nonneg}.
\begin{proof}
We prove the claim only for the LP rank,
as the proof for the SDP rank is similar.
We apply the Factorization Theorem
(Theorem~\ref{thm:factorization-nonneg}).
Let \(M_{\mathcal{P}_{1}}\) and \(M_{\mathcal{P}_{2}}\) denote the
slack matrices of \(\mathcal{P}_{1}\) and \(\mathcal{P}_{2}\),
respectively.
Then Eq.~\eqref{eq:red-nonneg} can be written as
\begin{equation}
  \label{eq:red-nonneg-complete_slack}
  M_{\mathcal{P}_{1}} =
  \left(
    F_{\mathfrak{I}}
    M_{\mathcal{P}_{2}}
    F_{\mathcal{S}}
  \right)
  \circ
  M_{1}
  +
  M_{2}
  ,
\end{equation}
where \(\circ\) denotes the Hadamard product (entrywise product),
and \(F_{\mathfrak{I}}\) and \(F_{\mathcal{S}}\) are the
\(\mathfrak{I}_{1} \times \mathfrak{I}_{2}\) and
\(\mathcal{S}_{2} \times \mathcal{S}_{1}\) matrices
encoding the two maps \(*\), respectively:
\begin{align}
  F_{\mathfrak{I}}(\mathcal{I}_{1}, \mathcal{I}_{2})
  &
  \coloneqq
  \begin{cases}
    1 & \text{if } \mathcal{I}_{2} = \mathcal{I}_{1}^{*}, \\
    0 & \text{if } \mathcal{I}_{2} \neq \mathcal{I}_{1}^{*};
  \end{cases}
  &
  F_{\mathcal{S}}(S_{2}, S_{1})
  &
  \coloneqq
  \begin{cases}
    1 & \text{if } S_{2} = S_{1}^{*}, \\
    0 & \text{if } S_{2} \neq S_{1}^{*}.
  \end{cases}
\end{align}
Let \(M_{\mathcal{P}_{2}} = \widetilde{M}_{\mathcal{P}_{2}}
+ a \allOne\)
with \(\LPrk M_{\mathcal{P}_{2}} =
\nnegrk \widetilde{M}_{\mathcal{P}_{2}}\).
This enables us to further simplify
Eq.~\eqref{eq:red-nonneg-complete_slack}:
\begin{equation}
  M_{\mathcal{P}_{1}} =
  \left(
    F_{\mathfrak{I}}
    \widetilde{M}_{\mathcal{P}_{2}}
    F_{\mathcal{S}}
  \right)
  \circ
  M_{1}
  +
  \operatorname{diag}(F_{\mathfrak{I}} a)
  \cdot
  M_{1}
  +
  M_{2}
  ,
\end{equation}
where \(\operatorname{diag}(x)\) stands for the square diagonal matrix
with the entries of \(x\) in the diagonal.
Now the claim follows from Theorem~\ref{thm:factorization-nonneg},
the well-known identities
\(\nnegrk (A \circ B) \leq \nnegrk A \cdot \nnegrk B\),
\(\nnegrk A B C \leq \nnegrk B\),
and the obvious \(\LPrk (A + B) \leq \LPrk A + \LPrk B\)
together with \(\LPrk(A B) \leq \LPrk B\).
\end{proof}
\end{theorem}

\subsection{LP and SDP formulations}
\label{sec:lp-sdp-formulations-1}

Here we recall the notion of linear programming and semi-definite programming
complexity of optimization problems from \cite{BPZ2015}.
The key idea to modeling approximations of an optimization problem
\(\mathcal{P} = (\mathcal{S}, \mathfrak{I}, \val)\)
is to represent the approximation gap by
two functions \(C, S \colon \mathfrak{I} \to \R\),
the \emph{completeness guarantee} and \emph{soundness guarantee},
respectively,
and the task is to differentiate problems
with \(\OPT{\mathcal{I}} \leq S(\mathcal{I})\)
and \(\OPT{\mathcal{I}} \geq C(\mathcal{I})\),
as in the algorithmic setting.

The guarantees \(C\) and \(S\) will often be
of the form \(C = \alpha g\) and \(S = \beta g\)
for some constants \(\alpha\) and \(\beta\)
and an easy-to-compute function \(g\).
Then we shall write
\(\fc(\mathcal{P}, \alpha, \beta)\)
instead of the more precise
\(\fc(\mathcal{P}, \alpha g, \beta g)\).

\begin{definition}[LP formulation of an optimization problem]
  \label{def:LP-formulation}
  Let
  \(\mathcal{P} = (\mathcal{S}, \mathfrak{I}, \val)\)
  be an optimization problem,
  and \(C, S\) be
  real-valued functions on \(\mathfrak{I}\),
  called \emph{completeness guarantee}
  and \emph{soundness guarantee}, respectively.
  If \(\mathcal{P}\) is a maximization problem,
  then let \(\mathfrak{I}^{S} \coloneqq
  \set{\mathcal{I} \in \mathfrak{I}}{\max
  \val_{\mathcal{I}} \leq S(\mathcal{I})}\)
  denote the set of instances, for which the soundness guarantee \(S\)
  is an upper bound on the maximum.
  If \(\mathcal{P}\) is a minimization problem,
  then let \(\mathfrak{I}^{S} \coloneqq
  \set{\mathcal{I} \in \mathfrak{I}}{\min \val_{\mathcal{I}}
   \geq S(\mathcal{I})}\)
  denote the set of instances,
  for which the soundness guarantee \(S\)
  is a lower bound on the minimum.

  A \emph{\((C, S)\)-approximate LP formulation} of \(\mathcal{P}\)
  consists of a linear program \(A x \leq b\)
  with \(x \in \R^{r}\) for some \(r\)
  and the following \emph{realizations}:
  \begin{description}
  \item[Feasible solutions] as vectors \(x^{s} \in \R^{r}\)
    for every \(s \in \mathcal{S}\) satisfying
  \begin{align}
    \label{eq:LP-contain}
    A x^{s} &\leq b \qquad \text{for all } s \in \mathcal{S},
  \end{align}
  i.e., the system \(Ax \leq b\) is a relaxation of
  \(\conv{x^s \mid s \in \mathcal{S}}\).
  \item[Instances] as affine functions
    \(w_{\mathcal{I}} \colon \R^{r} \to \R\)
    for all \(\mathcal{I} \in \mathfrak{I}^{S}\)
    satisfying
    \begin{align}
      \label{eq:LP-linear}
      w_{\mathcal{I}}(x^{s}) & = \val_{\mathcal{I}}(s)
            \qquad \text{for all } s \in
      \mathcal{S},
    \end{align}
    i.e., the linearization \(w_{\mathcal{I}}\) of
    \(\val_{\mathcal{I}}\)
    is required to be exact on all \(x^s\)
    with \(s \in \mathcal{S}\).
  \item[Achieving \((C,S)\) approximation guarantee]
  by requiring
  \begin{align}
    \label{eq:LP-approx}
    \max \set{w_{\mathcal{I}}(x)}{A x \leq b} &\leq C(\mathcal{I})
    \qquad \text{for all } \mathcal{I} \in \mathfrak{I}^{S},
  \end{align}
  if \(\mathcal{P}\) is a maximization problem
  (and \(\min
  \set{w_{\mathcal{I}}(x)}{A x \leq b} \geq C(\mathcal{I})\)
  if \(\mathcal{P}\) is a
  minimization problem).
  \end{description}
  The \emph{size} of the formulation is the number of inequalities
  in \(A x \leq b\).
  Finally, the
  \((C, S)\)-approximate
  \emph{LP formulation complexity} \(\fc(\mathcal{P}, C, S)\)
  of \(\mathcal{P}\) is
  the minimal size of all its LP formulations.
\end{definition}

The definition of SDP formulations is similar.
\begin{definition}[SDP formulation of an optimization problem]
  \label{def:SDP-formulation}
  As in Definition~\ref{def:LP-formulation},
  let
  \(\mathcal{P} = (\mathcal{S}, \mathfrak{I}, \val)\)
  be an optimization problem
  and \(C, S\) be
  real-valued functions on \(\mathfrak{I}\),
  the \emph{completeness guarantee} and \emph{soundness guarantee}.
  Let \(\mathfrak{I}^{S} \coloneqq
  \set{\mathcal{I} \in \mathfrak{I}}{\max \val_{\mathcal{I}}
   \leq S(\mathcal{I})}\)
  if \(\mathcal{P}\) is a maximization problem,
  and let \(\mathfrak{I}^{S} \coloneqq
  \set{\mathcal{I} \in \mathfrak{I}}{\min \val_{\mathcal{I}}
  \geq S(\mathcal{I})}\)
  if \(\mathcal{P}\) is a minimization problem.

  A \((C,S)\)-approximate \emph{SDP formulation} of
  \(\mathcal{P}\)
  consists of a linear map \(\mathcal{A} \colon \symM^{r} \to \R^{k}\)
  and a vector \(b \in \R^{k}\)
  (i.e., a semidefinite program
  \(\set{X \in \SDP^{r}}{\mathcal{A}(X) = b}\))
  together with the following \emph{realizations} of \(\mathcal{P}\):
  \begin{description}
  \item[Feasible solutions] as vectors \(X^{s} \in \SDP^{r}\)
    for all \(s \in \mathcal{S}\) satisfying
  \begin{align}
    \label{eq:SDP-contain}
    \mathcal{A}(X^{s}) = b
  \end{align}
  i.e., the SDP \(\mathcal{A}(X)  =  b, X \in \SDP^r\)
  is a relaxation of \(\conv{X^s}[s \in \mathcal{S}]\).
  \item[Instances] as affine functions
    \(w_{\mathcal{I}} \colon \symM^{r} \to \R\)
    for all \(\mathcal{I} \in \mathfrak{I}^{S}\)
    satisfying
    \begin{equation}
      \label{eq:SDP-linear}
      w_{\mathcal{I}}(X^{s}) = \val_{\mathcal{I}}(s)
            \qquad \text{for all } s \in
      \mathcal{S},
    \end{equation}
    i.e., the linearization \(w_{\mathcal{I}}\) of
    \(\val_{\mathcal{I}}\)
    is exact on the \(X^{s}\) with \(s \in \mathcal{S}\).
  \item[Achieving \((C,S)\) approximation guarantee]
  by requiring
  \begin{align}
    \label{eq:SDP-approx}
    \max \set{w_{\mathcal{I}}(X)}{\mathcal{A}(X^{s})  =  b,
    \ X^s \in \SDP^r}
    \leq C(\mathcal{I})
    \qquad \text{for all } \mathcal{I} \in \mathfrak{I}^S,
  \end{align}
  if \(\mathcal{P}\) is a maximization problem,
  and the analogous inequality
  if \(\mathcal{P}\) is a minimization problem.
  \end{description}

  The \emph{size} of the formulation is the dimension parameter \(r\).
  Now the \((C, S)\)-approximate \emph{SDP formulation complexity}
  \(\fcSDP(\mathcal{P}, C, S)\) of the
  problem \(\mathcal{P}\) is the minimal size of all its SDP formulations.
\end{definition}

\subsubsection{Slack matrix and formulation complexity}
\label{sec:slack-complexity}

The \((C, S)\)-approximate complexity of
a maximization problem
\(\mathcal{P} = (\mathcal{S}, \mathfrak{I}, \val)\)
is the complexity of proofs
of \(\val_{\mathcal{I}} \leq C(\mathcal{I})\)
for instances with \(\max \val_{\mathcal{I}} \leq S(\mathcal{I})\),
and similarly for minimization problems.
Formally, the proof complexity of the nonnegativity problem
\(\mathcal{P}_{C, S} = (\mathcal{S}, \mathfrak{I}^{S}, C - \val)\)
equals the \((C, S)\)-approximate complexity of \(\mathcal{P}\)
both in the LP and SDP world, as obvious from the definitions:
\begin{align}
  \fc(\mathcal{P}, C, S)
  &
  =
  \fc(\mathcal{P}_{C, S})
  ,
  &
  \fcSDP(\mathcal{P}, C, S)
  &
  =
  \fcSDP(\mathcal{P}_{C, S})
  .
\end{align}
Thus the theory of nonnegativity problems
from Section~\ref{sec:nonnegative-problem}
immediately applies,
which we formulate now explicitly for optimization problems.
The material here already appeared in \cite{BPZ2015}
without using nonnegativity problems and a significantly weaker reduction
mechanism.

The main technical tool for establishing lower bounds on the formulation complexity of a problem is its slack matrix
 and its factorizations
(decompositions). We start by recalling the definition of the slack
matrix for optimization problems.

\begin{definition}
  \label{def:slack-matrix}
  Let
  \(\mathcal{P} = (\mathcal{S}, \mathfrak{I}, \val)\)
  be an optimization problem with
 \emph{completeness guarantee} \(C\)
  and \emph{soundness guarantee} \(S\).
  The \emph{\((C, S)\)-approximate slack matrix}
  \(M_{\mathcal{P}, C, S}\) is the nonnegative
  \(\mathfrak{I}^{S} \times \mathcal{S}\) matrix
  with entries
  \begin{equation}
    \label{eq:slack-matrix}
    M_{\mathcal{P}, C, S}(\mathcal{I}, s) \coloneqq
    \tau \cdot (C(\mathcal{I}) - \val_{\mathcal{I}}(s))
    ,
  \end{equation}
  where \(\tau = +1\) if \(\mathcal{P}\) is a maximization problem,
  and \(\tau = -1\) if  \(\mathcal{P}\) is a minimization problem.
\end{definition}

Finally, we are ready to recall the factorization theorem,
equating LP rank and SDP rank with LP formulation complexity and SDP
formulation complexity, respectively.
The notion of LP and SDP rank is recalled in
Definition~\ref{def:matrix-factorization-nonneg}.
\begin{theorem}[Factorization theorem, \cite{BPZ2015}]
  \label{thm:factorization}
  Let
  \(\mathcal{P} = (\mathcal{S}, \mathfrak{I}, \val)\)
  be an optimization problem with
 \emph{completeness guarantee} \(C\)
  and \emph{soundness guarantee} \(S\).
  Then
  \begin{align}
    \label{eq:factorization_LP}
    \fc(\mathcal{P}, C, S)
    &
    = \LPrk M_{\mathcal{P}, C, S}
    ,
    \\
    \label{eq:factorization_SDP}
    \fcSDP(\mathcal{P}, C, S)
    &
    = \SDPrk M_{\mathcal{P}, C, S}
    ,
  \end{align}
  where \(M_{\mathcal{P}, C, S}\) is
  the \((C, S)\)-approximate slack matrix of \(\mathcal{P}\).
\end{theorem}
Now Theorem~\ref{thm:red-simple} follows as a special case
of Theorem~\ref{thm:red-nonneg}.

\subsubsection{Lasserre or SoS hierarchy}
\label{sec:lasserre-or-sos}

The Lasserre hierarchy,
also called the Sum-of-Squares (SoS) hierarchy,
is a series of SDP formulations of an optimization problem,
relying on a set of base functions.  The base functions are usually
chosen so that the objectives \(\val_{\mathcal{I}}\) of instances
are low-degree polynomials of the base functions.
For brevity, we recall only the optimal bound obtained by the
SDP formulation, using the notion of pseudoexpectation,
which is essentially a feasible point of the SDP. We follow the
definition of \cite[Page~3]{lee2014power}.

\begin{definition}[Lasserre/SoS hierarchy]
  \label{def:pseudo-expectation}
  \mbox{}
  \begin{description}
  \item[Pseudoexpectation]
    Let \(\{f_{1}, \dotsc, f_{\ell}\}\) be real-valued functions with
    common domain \(\mathcal{S}\).
    A \emph{pseudoexpectation functional} \(\pseudoexpectation\)
    of level \(d\)
    over \(\{f_{1}, \dotsc, f_{\ell}\}\)
    is a real-valued function with domain the vector space \(V\)
    of real-valued functions \(F\) with domain \(\mathcal{S}\),
    which are polynomials in \(f_{1}\), \dots, \(f_{\ell}\)
    of degree at most \(d\).
    A pseudoexpectation \(\pseudoexpectation\) is required to satisfy
    \begin{description}
    \item[Linearity]
      For all \(F_{1}, F_{2} \in V\)
      \begin{equation}
        \pseudoexpectation(F_{1} + F_{2}) =
        \pseudoexpectation(F_{1}) + \pseudoexpectation(F_{2})
        ,
      \end{equation}
      and for all \(r \in \R\) and \(F \in V\)
      \begin{equation}
        \pseudoexpectation(r F) = r \pseudoexpectation(F)
      \end{equation}
    \item[Positivity] \(\pseudoexpectation(F^{2}) \geq 0\)
      for all \(F \in V\) with degree at most \(d/2\)
      (so that \(F^{2} \in V\))
    \item[Normalization] \(\pseudoexpectation(1) = 1\)
      for the constant function \(1\).
    \end{description}
  \item[Lasserre or SoS value]
    Given an optimization problem
    \(\mathcal{P} = (\mathcal{S}, \mathfrak{I}, \val)\)
    and base functions \(f_{1}\), \dots, \(f_{\ell}\)
    defined on \(\mathcal{S}\),
    the \emph{degree \(d\) SoS value} or
    \emph{round \(d\) Lasserre value} of
    an instance \(\mathcal{I} \in \mathfrak{I}\) is
    \begin{equation}
      \SOS_{d}(\mathcal{I}) \coloneqq
      \max_{\pseudoexpectation
        \colon \deg \pseudoexpectation \leq 2 d}
      \pseudoexpectation(\val_{\mathcal{I}})
      .
    \end{equation}
  \end{description}
\end{definition}
Note that the base functions \(f_{i}\) might satisfy
non-trivial polynomial relations, and therefore
the vector space \(V\) need not be isomorphic
to the vector space of \emph{formal} low-degree polynomials
in the \(f_{i}\).
For example, if the \(f_{i}\) are all \(0/1\)-valued, which is a
common case, then \(f_{i}^{2}\) and \(f_{i}\) are
the same elements of \(V\).
We would also like to mention that
the degree or level \(d\) is not used consistently in the
literature, some papers use \(2d\) instead of our \(d\).
This results in a constant factor difference in the level,
which is usually not significant.

For CSPs we shall use the usual set of base functions
\(X_{x_{i} = \alpha}\),
the indicators that a variable \(x_{i}\)
is assigned the value \(\alpha\).
For graph problems, the solution set \(\mathcal{S}\)
usually consists of vertex sets or edge sets.  Therefore the common
choice of base functions are the indicators \(X_{v}\)
that a vertex or edge \(v\) lies in a solution.
This has been used for \UniqueGames in
\cite{Appr-Lasserre-UG_2010}
establishing an \(\omega(1)\) integrality gap
for an approximate Lasserre hierarchy
after a constant number of rounds.

\subsection{Base hard problems}
\label{sec:further-base-hard}

In this section we will recall the LP-hardness of the problems that
will serve as the starting point in our later reductions. We start
with the LP-hardness of the \Matching problem
with an inapproximability gap of \(1 - \varepsilon /n\):
\begin{theorem}[\cite{BP2014matchingJour},
  c.f., \cite{rothvoss2013matching}]
  \label{thm:mathRot} Let
  \(n \in \N\) and \(0 \leq \varepsilon < 1\).
  \begin{equation}
    \fc\left(
      \Matching{K_{2n}},
      \left\lfloor
        \frac{\size{V(H)}}{2}
      \right\rfloor
      +
      \frac{1 - \varepsilon}{2}
      ,
      \OPT{H}
    \right)
    = 2^{\Theta(n)},
  \end{equation}
  where \(H\) is the placeholder for the instance,
  and the constant factor in the exponent depends on \(\varepsilon\).
\end{theorem}

The following integrality gap was shown in
\cite{chan2013approximate} using the \MaxCUT
 Sherali-Adams integrality
gap instances of \cite{charikar2009integrality}.

\begin{theorem}[{\cite[Theorem~3.2]{chan2013approximate}}]
\label{thm:LP-MaxCUT}
For any \(\varepsilon > 0\) there are infinitely many \(n\) such
that
\begin{align*}
\fc\left(\MaxCUT{n}, 1 - \varepsilon, \frac{1}{2} +
\frac{\varepsilon}{6}\right) \ge n^{\Omega\left(\log{n}/\log\log{n}\right)}
\end{align*}
\end{theorem}

We now recall the Lasserre integrality gap result for
approximating \MaxCSP{k} from \cite{bhaskara2012polynomial}.
See also \cite{BGHMT2006,alekhnovich2005towards,STT2007,%
schoenebeck2008linear,tulsiani2009csp} for related results.

\begin{theorem}[{\cite[Theorem~4.2]{bhaskara2012polynomial}}]
  \label{thm:hardness-k-CSP}
For \(q \geq 2\), \(\varepsilon, \kappa > 0\) and \(\delta \geq 3/2\)
and large enough \(n\) depending on \(\varepsilon, \kappa, \delta\) and
\(q\), for every \(k, \beta\) satisfying
\(k \le n^{1/2}\) and \(\left(6q^k\ln{q}\right)/\varepsilon^2 \le \beta \le
n^{(1-\kappa)(\delta - 1)}/\left(10^{8(\delta-1)}k^{2\delta + 0.75}\right)\)
there is a \(k\)-ary predicate \(P \colon [q]^k \rightarrow \{0,1\}\)
 and a \MaxCSP{k}{P}
 instance \(\mathcal{I}\) on alphabet \([q]\) with \(n\) variables and \(m = \beta n\)
  constraints such that
  \(\OPT{\mathcal{I}}
\le O\left(\frac{1+\varepsilon}{q^k}\right)\), but the \(\frac{n\eta}{16}\)
round Lasserre relaxation for \(\mathcal{I}\) admits a perfect solution
with parameter
 \(\eta = \left. 1 \middle/
   \left(10^8(\beta k^{2\delta + 0.75})^{\frac{1}{\delta - 1}}\right)\right.\).
In other words, \(\SOS_{\eta n/16}(\mathcal{I}) = 1.\)
\end{theorem}

The following LP-hardness for \UniqueGames was shown in 
\cite{lee2014lower}
(based on \cite{chan2013approximate,charikar2009integrality}):
\begin{theorem}[{\cite[Corollary~7.7]{lee2014lower}}]
  \label{thm:hardnessofug}
 For every \(q \ge 2, \delta > 0\) and \(k \ge 1\) there exists a constant \(c > 0\) such that for all \(n \ge 1\)
 \[
 \fc\left(\UniqueGames{n}{q}, 1 - \delta, \frac{1}{q} + \delta\right) \ge cn^k
 \]
In other words there is no polynomial sized linear program
that approximates \UniqueGames within a factor of \(1 / q\).
\end{theorem}

\section{Reductions with distortion}
\label{sec:reductions}

We now introduce a generalization of the affine reduction
mechanism for LPs and SDPs as introduced in \cite{BPZ2015},
answering an open question posed both in \cite{BPZ2015,bazzi2015no},
leading to many new reductions that
were impossible in the affine framework.

\begin{definition}[Reduction]
  \label{def:red-simple}
  Let \(\mathcal{P}_{1} = (\mathcal{S}_{1}, \mathfrak{I}_{1}, \val)\)
  and
  \(\mathcal{P}_{2} = (\mathcal{S}_{2}, \mathfrak{I}_{2}, \val)\) be
  optimization problems with guarantees
  \(C_{1}, S_{1}\) and \(C_{2}, S_{2}\), respectively.
  Let \(\tau_{1} = +1\) if \(\mathcal{P}_{1}\) is a maximization
  problem, and \(\tau_{1} = -1\) if \(\mathcal{P}_{1}\) is a
  minimization problem.  Similarly, let \(\tau_{2} = \pm 1\)
  depending on whether \(\mathcal{P}_{2}\) is a maximization problem
  or a minimization problem.

  A \emph{reduction} from \(\mathcal{P}_{1}\)
  to \(\mathcal{P}_{2}\)
  respecting the guarantees
  consists of
  \begin{enumerate}
  \item
    two mappings:
    \(* \colon \mathfrak{I}_{1} \to \mathfrak{I}_{2}\)
    and
    \(* \colon \mathcal{S}_{1} \to \mathcal{S}_{2}\)
    translating instances and feasible solutions independently;
  \item
    two nonnegative \(\mathfrak{I}_{1} \times \mathcal{S}_{1}\)
    matrices \(M_{1}\), \(M_{2}\)
  \end{enumerate}
  subject to the conditions
  \begin{subequations}\label{eq:red-simple}
  \begin{align}
    %completeness
    \label{eq:red-simple-complete}
    \tau_{1}
    \left[
      C_{1}(\mathcal{I}_{1}) - \val_{\mathcal{I}_{1}}(s_{1})
    \right]
    &
    =
    \tau_{2}
    \left[
      C_{2}(\mathcal{I}_{1}^{*}) - \val_{\mathcal{I}_{1}^{*}}(s_{1}^{*})
    \right]
    M_{1}(\mathcal{I}_{1}, s_{1})
    +
    M_{2}(\mathcal{I}_{1}, s_{1})
    \tag{\theparentequation-complete}
    \\
    %soundness
    \label{eq:red-simple-sound}
    \tau_{2} \OPT{\mathcal{I}_{1}^{*}} &\leq \tau_{2}
    S_{2}(\mathcal{I}_{1}^{*})
    \qquad
    \text{if \(\tau_{1} \OPT{\mathcal{I}_{1}} \leq
    \tau_{1} S_{1}(\mathcal{I}_{1})\).}
    \tag{\theparentequation-sound}
  \end{align}
  \end{subequations}
\end{definition}
The matrices \(M_{1}\) and \(M_{2}\) provide extra freedom
to add additional (valid) inequalities during the reduction. In fact,
we might think of them as modeling more complex reductions. 
These matrices should have low computational overhead,
which in our framework means LP or SDP rank, as will be obvious from
the following special case of Theorem~\ref{thm:red-nonneg},
see Section~\ref{sec:lp-sdp-formulations-1} for details.

\begin{theorem}
  \label{thm:red-simple}
  Let \(\mathcal{P}_{1}\) and \(\mathcal{P}_{2}\) be optimization
  problems with a reduction from \(\mathcal{P}_{1}\)
  to \(\mathcal{P}_{2}\) respecting the
  completeness guarantees \(C_{1}\), \(C_{2}\)
  and soundness guarantees \(S_{1}\), \(S_{2}\)
  of \(\mathcal{P}_{1}\) and \(\mathcal{P}_{2}\), respectively.
  Then
  \begin{align}
    \fc(\mathcal{P}_{1}, C_{1}, S_{1})
    &
    \leq
    \LPrk M_{2} +
    \LPrk M_{1}
    + \nnegrk M_{1} \cdot \fc(\mathcal{P}_{2}, C_{2}, S_{2}),
    \\
    \fcSDP(\mathcal{P}_{1}, C_{1}, S_{1})
    &
    \leq
    \SDPrk M_{2} +
    \SDPrk M_{1}
    + \psdrk M_{1} \cdot \fcSDP(\mathcal{P}_{2}, C_{2}, S_{2}),
  \end{align}
  where \(M_{1}\) and \(M_{2}\) are the matrices in the reduction
  as in Definition~\ref{def:red-simple}.
\end{theorem}

The corresponding multiplicative inapproximability factors can be
obtained as usual, by taking the ratio of soundness and completeness.

\section{Fractional optimization problems}
\label{sec:fract-opt-problems-1}

A \emph{fractional optimization problem}
is an optimization problem where the objectives
have the form of a fraction
\(\val_{\mathcal{I}} = \val^{n}_{\mathcal{I}} /
\val^{d}_{\mathcal{I}}\),
such as for \SparsestCut.
In this case the affine space of the objective functions
\(\val_{\mathcal{I}}\)
of instances is typically not low dimensional,
immediately ruling out small linear and
semidefinite formulations.
Nevertheless,
there are examples of efficient linear programming based algorithms
for such problems, however here the linear programs are used to find an
optimal value of a linear combination of \(\val^{n}_{\mathcal{I}}\)
and \(\val^{d}_{\mathcal{I}}\) (see e.g., \cite{gupta2013sparsest}). 
To be able to analyze the size of LPs or SDPs for such problems
we refine the notion of formulation complexity from \cite{BPZ2015}
to incorporate these types of linear programs,
which reduces to the original definition
with the choice of
\(\val^{n}_{\mathcal{I}} = \val_{\mathcal{I}}\) and
\(\val^{d}_{\mathcal{I}} = 1\).

We now provide the formal definitions of linear programming and
semidefinite formulations for fractional optimization problems.
The idea is again that the complexity is essentially the proof
complexity of \(\val_{\mathcal{I}} \leq C(\mathcal{I})\)
for instances with \(\val_{\mathcal{I}} \leq S(\mathcal{I})\).
Formally, given a fractional optimization problem
\(\mathcal{P} = (\mathcal{S}, \mathfrak{I}, \val)\)
with guarantees \(C, S\),
we study the nonnegativity problem
\(\mathcal{P}_{C, S} = (\mathcal{S}, \mathfrak{I}^{S} \times \{0, 1\},
\val^{*})\)
with \(\val^{*}_{(\mathcal{I}, 0)} =
C(\mathcal{I}) \val^{d}_{\mathcal{I}} - \val^{n}_{\mathcal{I}}\)
(encoding \(\val_{\mathcal{I}} \leq C(\mathcal{I})\))
and \(\val^{*}_{(\mathcal{I}, 1)} = \val^{d}_{\mathcal{I}}\).
The addition of \(\val^{d}\)
to the objective functions is for the technical reason to ensure that
the objectives span the same affine space as the
\(\val^{n}_{\mathcal{I}}\)
and \(\val^{d}_{\mathcal{I}}\),
i.e., to capture the affineness of these functions.  This is not
expected to significantly affect the complexity of the resulting
problem, as the \(\val^{d}_{\mathcal{I}}\)
in interesting applications are usually a positive linear combination
of a small number of nonnegative functions.

As a special case of
Section~\ref{sec:nonnegative-problem} we obtain the following setup
for fractional optimization problems.
Note that when \(\mathcal{P}\) is a fractional optimization
problem with \(\val^d = 1\), then \(\mathcal{P}\) is an optimization
problem and
Definitions~\ref{def:LP-formulation-fractional}
and~\ref{def:SDP-formulation-fractional}
are equivalent to Definitions~\ref{def:LP-formulation}
and~\ref{def:SDP-formulation},
as we will see now.

\begin{definition}[LP formulation of
  a fractional optimization problem]
\label{def:LP-formulation-fractional}
Let \(\mathcal{P} = (\mathcal{S}, \mathfrak{I}, \val)\)
be a fractional optimization problem
and let \(C, S\) be
two real valued functions on \(\mathfrak{I}\) called
\emph{completeness guarantee} and \emph{soundness guarantee}
respectively. Let
\(\mathfrak{I}^S \defeq \{\mathcal{I} \in \mathfrak{I} \mid
\max \val_{\mathcal{I}} \le S(\mathcal{I})\}\) when \(\mathcal{P}\)
is a maximization problem and \(\mathfrak{I}^S \defeq
\{\mathcal{I} \in \mathfrak{I} \mid \min \val_{\mathcal{I}} \ge
S(\mathcal{I})\}\) if \(\mathcal{P}\) is a minimization problem.

A \emph{\((C, S)\)-approximate LP formulation} for the problem
\(\mathcal{P}\) consists of a linear program \(Ax \le b\) with
\(x \in \R^r\) for some \(r\) and
the following realizations:

\begin{description}
\item[Feasible solutions] as vectors \(x^s \in
\R^r\) for every \(s \in \mathcal{S}\) satisfying
\begin{align*}
Ax^s \le b \qquad \text{ for all } s \in \mathcal{S},
\end{align*}
i.e., \(Ax \le b\) is a relaxation of \(\conv{x^s \mid s \in
\mathcal{S}}\).

\item[Instances] as a pair of affine functions
\(w^n_{\mathcal{I}}, w^d_{\mathcal{I}} \colon \R^r \rightarrow
\R\) for all \(\mathcal{I} \in \mathfrak{I}^S\)
 satisfying
\begin{align*}
w^n_{\mathcal{I}}(x^s) = \val^n_{\mathcal{I}}(s)\\
w^d_{\mathcal{I}}(x^s) = \val^d_{\mathcal{I}}(s)
\end{align*}
for every \(s\in \mathcal{S}\). In other words the linearizations
\(w^n_{\mathcal{I}}, w^d_{\mathcal{I}}\) are required to be
\emph{exact} on all \(x^s\) for \(s \in \mathcal{S}\).

\item[Achieving \((C,S)\) approximation guarantee] requiring the following for
every \(\mathcal{I} \in \mathfrak{I}^S\)
\begin{align*}
Ax \le b \Rightarrow \begin{cases} w^d_{\mathcal{I}}(x) \ge 0\\
		w^n_{\mathcal{I}}(x) \le C(\mathcal{I})w^d_{\mathcal{I}}(x)
		\end{cases}
\end{align*}
if \(\mathcal{P}\) is a maximization problem and
\begin{align*}
Ax \le b \Rightarrow \begin{cases} w^d_{\mathcal{I}}(x) \ge 0\\
		w^n_{\mathcal{I}}(x) \ge C(\mathcal{I})w^d_{\mathcal{I}}(x)
		\end{cases}
\end{align*}
if \(\mathcal{P}\) is a minimization problem. In other words
we can derive the nonnegativity of \(w^d_{\mathcal{I}}\)
and the approximation guarantee \(C(\mathcal{I})\) from the set of
inequalities in \(Ax \le b\).

\end{description}
The \emph{size} of the formulation is the number of inequalities
in \(Ax \le b\). Finally, the \((C, S)\)-approximate
\emph{LP formulation complexity} \(\fc{(\mathcal{P}, C, S)}\) of
\(\mathcal{P}\) is the minimal size of all its LP formulations.
\end{definition}

SDP formulations for fractional optimization problems are defined
similarly.

\begin{definition}[SDP formulation of fractional optimization problem]
\label{def:SDP-formulation-fractional}
Let \(\mathcal{P} = (\mathcal{S}, \mathfrak{I}, \val)\) be a fractional
optimization problem and let \(C, S \colon \mathfrak{I} \rightarrow
\R_{\ge 0}\) be the \emph{completeness guarantee} and the
\emph{soundness guarantee} respectively.
Let
\(\mathfrak{I}^S \defeq \{\mathcal{I} \in \mathfrak{I} \mid
\max \val_{\mathcal{I}} \le S(\mathcal{I})\}\) when \(\mathcal{P}\)
is a maximization problem and \(\mathfrak{I}^S \defeq
\{\mathcal{I} \in \mathfrak{I} \mid \min \val_{\mathcal{I}} \ge
S(\mathcal{I})\}\) if \(\mathcal{P}\) is a minimization problem.

A \((C, S)\)-approximate SDP formulation of \(\mathcal{P}\) consists
of a linear map \(\mathcal{A} \colon \symM^r \rightarrow \R^k\)
together
with a vector \(b \in \R^k\) (i.e., a semidefinite program
\(\{X \in \symM^r_+ \mid \mathcal{A}(X) = b\}\)) and the following
realizations of \(\mathcal{P}\):
\begin{description}
\item[Feasible solutions] as vectors \(X^s \in \symM^r_+\)
for every \(s \in \mathcal{S}\) satisfying
\begin{align*}
\mathcal{A}(X^s) = b \qquad \text{ for every } s \in \mathcal{S},
\end{align*}
i.e. \(\mathcal{A}(X) = b, X \in \SDP^r\) is a relaxation of
\(\conv{X^s \mid s \in \mathcal{S}}\).

\item[Instances] as a pair of affine functions
\(w^n_{\mathcal{I}}, w^d_{\mathcal{I}} \colon \symM^r \rightarrow
\R_{\ge 0}\) for every \(\mathcal{I} \in \mathfrak{I}^S\)
satisfying
\begin{align*}
w^n_{\mathcal{I}}(X^s) = \val^n_{\mathcal{I}}(s)\\
w^d_{\mathcal{I}}(X^s) = \val^d_{\mathcal{I}}(s)
\end{align*}
for every \(s \in \mathcal{S}\). In other words the linearizations
\(w^n_{\mathcal{I}}, w^d_{\mathcal{I}}\) are required to be
\emph{exact} on all \(X^s\) for \(s \in \mathcal{S}\).

\item[Achieving \((C,S)\) approximation guarantee] requiring the following for
every \(\mathcal{I} \in \mathfrak{I}^S\)
\begin{align*}
\mathcal{A}(X) = b \Rightarrow \begin{cases} w^d_{\mathcal{I}}(X)
\ge 0 \\
w^n_{\mathcal{I}}(X) \le C(\mathcal{I})w^d_{\mathcal{I}}(X)
\end{cases}
\end{align*}
if \(\mathcal{P}\) is a maximization problem and
\begin{align*}
\mathcal{A}(X) = b \Rightarrow \begin{cases} w^d_{\mathcal{I}}(X)
\ge 0 \\
w^n_{\mathcal{I}}(X) \le C(\mathcal{I})w^d_{\mathcal{I}}(X)
\end{cases}
\end{align*}
if \(\mathcal{P}\) is a minimization problem.
\end{description}
The \emph{size} of the formulation is given by the dimension \(r\).
The \((C, S)\)-approximate \emph{SDP formulation complexity}
\(\fcSDP(\mathcal{P}, C, S)\) of the problem \(\mathcal{P}\)
is the minimal size of all its SDP formulations.
\end{definition}

The slack matrix for fractional problems plays the same role
as for non-fractional problems, with the twist that we factorize
the denominator and numerator separately. This allows us to overcome
the high dimensionality of the space spanned by the actual
ratios. 

\begin{definition}\label{def:slack-fractional}
Let \(\mathcal{P} = (\mathcal{S}, \mathfrak{I}, \val)\) be
a fractional optimization problem with \emph{completeness guarantee}
\(C\) and \emph{soundness guarantee} \(S\). The
\emph{\((C, S)\)-approximate slack matrix} \(M_{\mathcal{P}, C, S}\)
is the nonnegative \(2\mathfrak{I}^{S} \times \mathcal{S}\) matrix
of the form
\begin{align*}
  M_{\mathcal{P}, C, S} = \begin{bmatrix}
M^{(d)}_{\mathcal{P}, C, S} \\
M^{(n)}_{\mathcal{P}, C, S}
\end{bmatrix}
\end{align*}
where \(M^{(d)}_{\mathcal{P}, C, S}, M^{(n)}_{\mathcal{P}, C, S}\) are
nonnegative \(\mathfrak{I}^S\times \mathcal{S}\) matrices with
entries
\begin{align*}
M^{(d)}_{\mathcal{P}, C, S}(\mathcal{I}, s) &\defeq \val^d_{\mathcal{I}}(s)\\
M^{(n)}_{\mathcal{P}, C, S}(\mathcal{I}, s) &\defeq
\tau\left(C(\mathcal{I})\val^d_{\mathcal{I}}(s) -
\val^n_{\mathcal{I}}(s)\right)
\end{align*}
where \(\tau = +1\) if \(\mathcal{P}\) is a maximization problem
and \(\tau = -1\) if \(\mathcal{P}\) is a minimization problem.
\end{definition}

We are now ready to obtain the factorization theorem for the
class of fractional optimization problems,
as a special case of Theorem~\ref{thm:factorization-nonneg}:

\begin{theorem}[Factorization theorem for fractional optimization
problems]\label{thm:factorization-fractional}
Let \(\mathcal{P} = (\mathcal{S}, \mathfrak{I}, \val)\) be a
fractional optimization problem with \emph{completeness
guarantee} \(C\) and \emph{soundness guarantee} \(S\). Then
\begin{align*}
\fc(\mathcal{P}, C, S) &= \LPrk M_{(\mathcal{P}, C, S)},\\
\fcSDP(\mathcal{P}, C, S) &= \SDPrk M_{(\mathcal{P}, C, S)}
\end{align*}
where \(M_{(\mathcal{P}, C, S)}\) is the \((C, S)\)-approximate
slack matrix of \(\mathcal{P}\).
\end{theorem}
Now Theorem~\ref{thm:red-fraction} arises as a special case
of Theorem~\ref{thm:red-nonneg}.

\subsection{Reduction between fractional problems}
\label{sec:reduct-fractional}

Reductions for fractional optimization problems are completely
analogous
to the non-fractional case:
\begin{definition}[Reduction]
  \label{def:red-fraction}
  Let \(\mathcal{P}_{1} = (\mathcal{S}_{1}, \mathfrak{I}_{1}, \val)\)
  and
  \(\mathcal{P}_{2} = (\mathcal{S}_{2}, \mathfrak{I}_{2}, \val)\) be
  fractional
  optimization problems with guarantees
  \(C_{1}, S_{1}\) and \(C_{2}, S_{2}\), respectively.
  Let \(\tau_{1} = +1\) if \(\mathcal{P}_{1}\) is a maximization
  problem, and \(\tau_{1} = -1\) if \(\mathcal{P}_{1}\) is a
  minimization problem.  Similarly, let \(\tau_{2} = \pm 1\)
  depending on whether \(\mathcal{P}_{2}\) is a maximization problem
  or a minimization problem.

  A \emph{reduction} from \(\mathcal{P}_{1}\)
  to \(\mathcal{P}_{2}\)
  respecting the guarantees
  consists of
  \begin{enumerate}
  \item
    two mappings:
    \(* \colon \mathfrak{I}_{1} \to \mathfrak{I}_{2}\)
    and
    \(* \colon \mathcal{S}_{1} \to \mathcal{S}_{2}\)
    translating instances and feasible solutions independently;
  \item
    four nonnegative \(\mathfrak{I}_{1} \times \mathcal{S}_{1}\)
    matrices \(M^{(n)}_{1}\), \(M^{(d)}_{1}\),
    \(M_{2}^{(n)}\), \(M^{(d)}_{2}\)
  \end{enumerate}
  subject to the conditions
  \begin{subequations}\label{eq:red-fraction}
  \begin{align}
    %completeness
    \tau_{1}
    \left[
      C_{1}(\mathcal{I}_{1}) \val^{d}_{\mathcal{I}_{1}}(s_{1})
      - \val^{n}_{\mathcal{I}_{1}}(s_{1})
    \right]
    \label{eq:red-fraction-complete}
    &
    =
    \tau_{2}
    \left[
      C_{2}(\mathcal{I}_{1}^{*}) \val^{d}_{\mathcal{I}^*_{1}}(s^*_{1})
      - \val^{n}_{\mathcal{I}_{1}^{*}}(s_{1}^{*})
    \right]
    M_{1}^{(n)}(\mathcal{I}_{1}, s_{1})
    +
    M_{2}^{(n)}(\mathcal{I}_{1}, s_{1})
    \tag{\theparentequation-complete}
    \\
    \val^{d}_{\mathcal{I}_{1}}(s_{1})
    &
    =
    \val^{d}_{\mathcal{I}_{1}^{*}}(s_{1}^{*})
    \cdot
    M^{(d)}_{1}(\mathcal{I}_{1}, s_{1})
    +
    M^{(d)}_{2}(\mathcal{I}_{1}, s_{1})
    \tag{\theparentequation-denominator}
    \\
    %soundness
    \label{eq:red-fraction-sound}
    \tau_{2} \OPT{\mathcal{I}_{1}^{*}} &\leq \tau_{2}
    S_{2}(\mathcal{I}_{1}^{*})
    \qquad
    \text{if \(\tau_{1} \OPT{\mathcal{I}_{1}} \leq
    \tau_{1} S_{1}(\mathcal{I}_{1})\).}
    \tag{\theparentequation-sound}
  \end{align}
  \end{subequations}
\end{definition}
As the \(\val^{d}\) are supposed to have a small proof,
the matrices \(M^{(d)}_{1}\) and \(M^{(d)}_{2}\)
are not supposed to significantly influence the strength of
the reduction even with the trivial choice \(M^{(d)}_{1} = 0\) and
\(M^{(d)}_{2}(\mathcal{I}_{1}, s_{1}) = \val^{d}_{\mathcal{I}_{1}}(s_{1})\).
However, as in the non-fractional case,
the complexity of \(M^{(n)}_{1}\) and \(M^{(n)}_{2}\)
could have a major influence on the strength of the reduction.
The reduction theorem is a special case of
Theorem~\ref{thm:red-nonneg},
see Section~\ref{sec:fract-opt-problems-1}:
\begin{theorem}
  \label{thm:red-fraction}
  Let \(\mathcal{P}_{1}\) and \(\mathcal{P}_{2}\) be optimization
  problems with a reduction from \(\mathcal{P}_{1}\)
  to \(\mathcal{P}_{2}\)
  Then
  \begin{align}
    \fc(\mathcal{P}_{1}, C_{1}, S_{1})
    &
    \leq
    \LPrk
    \begin{bmatrix}
      M^{(n)}_{2}
      \\
      M^{(d)}_{2}
    \end{bmatrix}
    +
    \LPrk
    \begin{bmatrix}
      M^{(n)}_{1}
      \\
      M^{(d)}_{1}
    \end{bmatrix}
    + \nnegrk
    \begin{bmatrix}
      M^{(n)}_{1}
      \\
      M^{(d)}_{1}
    \end{bmatrix}
    \cdot
    \fc(\mathcal{P}_{2}, C_{2}, S_{2}),
    \\
    \fcSDP(\mathcal{P}_{1}, C_{1}, S_{1})
    &
    \leq
    \SDPrk
    \begin{bmatrix}
      M^{(n)}_{2}
      \\
      M^{(d)}_{2}
    \end{bmatrix}
    +
    \SDPrk
    \begin{bmatrix}
      M^{(n)}_{1}
      \\
      M^{(d)}_{1}
    \end{bmatrix}
    + \psdrk
    \begin{bmatrix}
      M^{(n)}_{1}
      \\
      M^{(d)}_{1}
    \end{bmatrix}
    \cdot
    \fcSDP(\mathcal{P}_{2}, C_{2}, S_{2}),
  \end{align}
  where \(M^{(n)}_{1}\), \(M^{(d)}_{1}\), \(M^{(n)}_{2}\),
  and \(M^{(d)}_{2}\) are the matrices in the reduction
  as in Definition~\ref{def:red-fraction}.
\end{theorem}

\section{A simple example: \Matching over \(3\)-regular graphs has no small LPs}
\label{sec:matching-over-3}

We now show that the \Matching problem even over
\(3\)-regular
graphs does not admit a small LP formulation.  This has been an open
question of various researchers,
given that the \Matching problem admits polynomial-size LPs
for many classes of sparse graphs, like bounded treewidth,
planar (and bounded genus) graphs
\cite{barahona1993cuts,gerards1991compact,Kolman15}.
We also show that for graphs of bounded degree \(3\), the \Matching
problem does not admit fully-polynomial size relaxation schemes,
the linear programming equivalent of FPTAS,
see
\cite{BP2014matching,BP2014matchingJour}
for details on these schemes.
\begin{theorem} \label{thm:3degLPhard} Let \(n \in
  \N\) and \(0 \leq \varepsilon < 1\).
  There exists a \(3\)-regular graph \(D_{2n}\) with \(2n (2n-1)\)
  vertices, so that 
  \begin{equation}
    \fc\left(
      \Matching{D_{2n}},
      \left\lfloor
        \frac{\size{V(H)}}{2}
      \right\rfloor
      +
      \frac{1 - \varepsilon}{2}
      ,
      \OPT{H}
    \right)
    = 2^{\Omega(\sqrt{\size{V(D_{2n})}})}
    ,
  \end{equation}
  where \(H\) is the placeholder for an instance,
  and the constant factor in the exponent depends on \(\varepsilon\).
  In particular, \Matching{D_{2n}} is LP-hard
  with an inapproximability factor of \(1 - \varepsilon / \size{V(D_{2n})}\).
\begin{proof}
As usual,
the inapproximability factor simply arises as the smallest factor
\(\left. \OPT{H} \middle/ \left( \lfloor \size{V(H)} / 2 \rfloor
    + (1 - \varepsilon)/2 \right) \right.\)
of the soundness and completeness guarantees.

  The proof is a simple application of the reduction framework. In
  fact, it suffices to use the affine framework of \cite{BPZ2015}. We
  will reduce from the perfect matching problem \(\Matching{K_{2n}}\)
  as given in Definition~\ref{prob:pm}.

We first construct our target graph \(D_{2n}\) as follows,
see Figure~\ref{fig:3-regular-matching}:
\begin{enumerate}
\item\label{item:1} For every vertex \(v\) of \(K_{2n}\)
  we consider a cycle \(C^{v}\) of length \(2n-1\).
  We denote the vertices of \(C^{v}\) by \([v,u]\),
  where \(v, u \in V\) and \(v \neq u\).
\item\label{item:2} The graph \(D_{2n}\) is the disjoint union of the
  \(C^{v}\)
  for \(v \in V\)
  together with the following additional edges:
  an edge \(([v,u],[u,v])\) for every \((u,v) \in E\).
\end{enumerate}
Thus \(D_{2n}\) has a total of \(2n (2n-1)\) vertices.
\begin{figure}[htb!]
  \centering
  \includegraphics{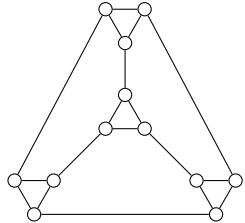}
  \caption{\label{fig:3-regular-matching}%
    The graph \(D_{2n}\) for \(n=2\)
    in the reduction to \(3\)-regular Matching.}
\end{figure}
This completes the definition of the graph \(D_{2n}\),
which is obviously \(3\)-regular.
(There is some ambiguity regarding the order of vertices in
the cycles \(C^{v}\), but this does not affect the argument below.)
Now we define the reduction from \(\Matching{K_{2n}}\)
to \(\Matching{D_{2n}}\).

We first map the instances.
Let \(H\) be a subgraph of \(K_{2n}\).
Its image \(H^*\) under the reduction
is the union of the \(C^{v}\) for \(v \in H\)
together with the edges \(([u,v], [v,u])\)
for \(\{u, v\} \in E(H)\).

Now let \(M\) be a perfect matching in \(K_{2n}\). We define \(M^*\)
by naturally extending it to a perfect matching in
\(D_{2n}\).
For every edge \(e = \{u, v\} \in M\) in the matching,
the edges \(([u,v],[v,u]) \in D_{2n}\)
form a matching containing exactly one vertex from every cycle
\(C^{v}\).
We choose \(M\) to be the unique extension of this matching
to a perfect matching by adding edges from the cycles \(C^{v}\).

We obviously have the following relationship between
the objective values:
\begin{equation}
 \begin{split}
  \val^{D_{2n}}_{H^{*}}(M^{*})
  &
  =
  \size{M^{*} \cap E(H^{*})}
  =
  \size{V(H)} \cdot (n - 1)
  +
  \size{M \cap V(H)}
  \\
  &
  =
  \size{V(H)} \cdot (n - 1)
  +
  \val^{K_{2n}}_{H}(M)
  ,
 \end{split}
\end{equation}
providing immediately the completeness of the reduction:
\begin{equation}
  \begin{split}
  \left\lfloor
    \frac{\size{V(H)}}{2}
  \right\rfloor
  +
  \frac{1 - \varepsilon}{2}
  -  \val^{K_{2n}}_{H}(M)
  &
  =
  \left\lfloor
    \frac{\size{V(H)} \cdot (2n - 1)}{2}
  \right\rfloor
  +
  \frac{1 - \varepsilon}{2}
  - \val^{D_{2n}}_{H^{*}}(M^{*}) \\
  &
  =
  \left\lfloor
    \frac{\size{V(H^*)}}{2}
  \right\rfloor
  +
  \frac{1 - \varepsilon}{2}
  -  \val^{D_{2n}}_{H^{*}}(M^{*})
  .    
  \end{split}
\end{equation}
The soundness of the reduction is immediate,
as the soundness guarantee is the optimal value.
\end{proof}
\end{theorem}

It is an interesting open problem, whether there exists a family of
bounded-degree graphs \(G_n\) on \(n\) vertices so that the lower bound in
Theorem~\ref{thm:3degLPhard} can be strengthened to \(2^{\Omega(n)}\).

\section{\Problem{BalancedSeparator} and \Problem{SparsestCut}}
\label{sec:lp-hardness-balanced}
The \SparsestCut problem is a high-profile problem that received
considerable attention in the past. It is known that \SparsestCut with 
general demands can be
approximated within a factor of \(O(\sqrt{\log
  n}\log\log n)\) \cite{arora2008euclidean} and that the standard SDP 
  has an integrality
gap of \((\log n)^{\Omega(1)}\) \cite{cheeger2009log}. 
The \Problem{BalancedSeparator} problem is a related 
problem which often
arises in connection to the \SparsestCut problem
(see Definition~\ref{def:balanced-separator}). 
The main result of this section will be to show that
the \SparsestCut and \Problem{BalancedSeparator}
 problems cannot be
approximated well by small
LPs and SDPs by using the new reduction mechanism from
Section~\ref{sec:reduct-fractional}.
In the case of the \SparsestCut problem our result holds
 even if the supply graph
has bounded treewidth, with the lower bound matching the upper bound
in \cite{gupta2013sparsest} in the LP case.
The results are \emph{unconditional}
LP/SDP analogues to \cite{chawla2006hardness},
however for a different regime.
In the case of the \Problem{BalancedSeparator}
 problem our result
holds even if the demand graph has bounded treewidth.

The \SparsestCut problem is a fractional optimization problem:
we extend Definition~\ref{def:sparsestCut} via
\begin{align}
  \val^n_{\mathcal{I}}(s) \defeq \sum_{i \in s, j \notin s} c(i, j), \qquad
  \val^d_{\mathcal{I}}(s) \defeq \sum_{i \in s, j \notin s} d(i, j)
\end{align}
for any vertex set \(s\) and any instance \(\mathcal{I}\)
with capacity \(c\) and demand \(d\).

\begin{theorem}[LP/SDP hardness for \SparsestCut,
  \(\tw(\text{supply})
= O(1)\)]
\label{thm:hardness-sc}
For any \(\varepsilon \in (0, 1)\) there are
\(\eta_{\textrm{LP}} > 0\) and \(\eta_{\textrm{SDP}} > 0\)
such that for every large enough \(n\) the following hold
\begin{align*}
  \fc\left(\SparsestCut{n}{2}, \eta_{\textrm{LP}}(1 + \varepsilon),
    \eta_{\textrm{LP}} \left(
2 -  \varepsilon\right)\right) &\ge 
n^{\Omega\left(\log n / \log \log n\right)},
  \\
  \fcSDP\left(\SparsestCut{n}{2}, \eta_{\textrm{SDP}}\left(1
      + \frac{4\varepsilon}{5}\right),\eta_{\textrm{SDP}}
    \left(\frac{16}{15} - \varepsilon\right)\right)
&\ge n^{\Omega\left(\log{n}/\log\log{n}\right)}.
\end{align*}
In other words \SparsestCut{n}{2} is LP-hard with an
inapproximability factor of \(2 - \varepsilon\),
and SDP-hard with an
inapproximability factor of \(\frac{16}{15} - O(\varepsilon)\).
\end{theorem}

A complementary reduction proves the hardness of approximating
\Problem{BalancedSeparator}
where the demand graph has constant treewidth.
Note that we only have an inapproximability result for LPs in
this case since the reduction is from \UniqueGames for which
we do not yet know of any SDP hardness result.

\begin{theorem}[LP-hardness for \Problem{BalancedSeparator}]
\label{thm:balSepHardness}
For any constant \(c_1 \ge 1\) there is another
constant \(c_2 \ge 1\) such that for all \(n\)
there is a demand function 
\(d \colon E(K_n) \rightarrow \R_{\ge 0}\) satisfying \(\tw([n]_d) \le c_2\)
so that \BalancedSeparator{n}{d} is LP-hard with an inapproximability
factor of \(c_1\).
\end{theorem}

\subsection{\Problem{SparsestCut} with bounded treewidth supply graph}\label{sec:sparsest-cut-supply}
In this section we show that the \SparsestCut problem over supply
graphs with treewidth \(2\) cannot be
approximated up to a factor of  
\(2\) by any polyonomial sized LP and up to a factor of
\(\frac{16}{15}\) by any polynomial sized SDP,
i.e., Theorem~\ref{thm:hardness-sc}.

We use the reduction from \cite{gupta2013sparsest},
reducing \MaxCUT to \SparsestCut.
Given an instance
\(\mathcal{I}\)
of \MaxCUT{n} we first construct the 
instance \(\mathcal{I}^*\) on vertex
set \(V = \{u, v\} \cup [n]\) where \(u\) and \(v\) are
two special vertices. Let us denote the degree
of a vertex \(i\) in \(\mathcal{I}\) by \(\deg(i)\)
and let \(m \defeq \frac{1}{2}\sum_{i=1}^n \deg(i)\) be the
total number of edges in \(\mathcal{I}\). 
We define the capacity function \(c \colon V \times V \rightarrow
\R_{\ge 0}\) as 
\begin{align*}
c(i, j) &\defeq \begin{cases} \frac{\deg(i)}{m}
 \qquad & \text{ if } j = u, i \neq v \text{ or } j = v, i \neq u\\
 0 \qquad & \text{ otherwise.}
 \end{cases}
\intertext{Note that the supply graph has treewidth at most \(2\)
being a copy of \(K_{2, n}\).
The demand function \(d \colon V \times V \rightarrow \R_{\ge 0}\)
is defined as 
}
d(i, j) &\defeq 
		   \begin{cases} 
		   \frac{2}{m} \qquad &\text{ if } \{i, j\}
		   \in E(\mathcal{I}) \\
		   0 \qquad &\text{ otherwise.}
		   \end{cases}
\end{align*}
We map a solution \(s\) to \MaxCUT{n} to the cut
\(s^* \defeq s \cup \{u\}\) of \SparsestCut{n+2}{2}.

We remind the reader of the powering operation from
\cite{gupta2013sparsest} to handle the case of
unbalanced and non \(u\)-\(v\) cuts. 
It successively adds for every edge of
\(\mathcal{I}^{*}\) a copy of itself, scaling both the capacities
and demands by the capacity of the edge.
After \(l\) rounds, we obtain an instance \(\mathcal{I}^*_{l}\)
on a fixed set of \(O(N^{2l})\) vertices,
and similarly the cuts \(s^{*}\) extend naturally to cuts
 \(s^*_{l}\)
on these vertices, independent of the instance \(\mathcal{I}\).
We provide a formal definition of the powering operation below,
for any general instance \(\mathcal{I}_1\) and general
solution \(s_1\) of \SparsestCut.

\begin{definition}[Powering instances]
\label{def:instance-powering}
  The instances of \SparsestCut{N_{1}}
  are \(G_{1} \coloneqq  K_{N_{1}}\) with capacity function \(c_1\) 
  and demand function \(d_1\).
  Let \(u\) and \(v\) be two distinguished vertices of \(G_{1}\).
  We construct a sequence \(\{G_{l}\}_{l}\) of graphs
  with distinguished vertices \(u\) and \(v\) recursively
  as follows.
  The graph \(G_{l}\) is obtained by replacing every edge
  \(\{x, y\}\) of \(G_{1}\) by a copy of \(G_{l-1}\).
  Let us denote by \((\{x,y\}, w)\) the copy of vertex \(w\)
  of \(G_{l-1}\).
  We identify the vertices \((\{x, y\}, u)\) and \((\{x,y\}, v)\)
  with \(x\) and \(y\). There are two ways to do so for every edge and
  we can pick either, arbitrarily.
  Obviously, \(G_{l}\) has
  \(N_{l} \coloneqq \sum_{i=1}^{l-1} \binom{N}{2}^{i} (N - 2) + 2\)
  many vertices.
Given a base instance \(\mathcal{I}_1\) of 
\SparsestCut{N_1} we will construct a sequence of instances
\(\{\mathcal{I}_l\}_{l}\) of \SparsestCut{N_{l}} recursively
as follows. Let the capacity and demand function of
\(\mathcal{I}_{l-1}\) be
\(c_{l-1}\) and \(d_{l-1}\) respectively.
The capacity of edges not in \(G_{l}\) will be \(0\).
Any edge \(e\) of \(G_{l}\) has the unique form
\(\{(\{x, y\}, p), (\{x, y\}, q)\}\)
for an edge \(\{x, y\}\) of \(G_{1}\) and an edge \(\{p, q\}\)
of \(G_{l-1}\).
We define
\(c_l(e) \defeq c_{l-1}(p, q) \cdot c_1(x, y)\).
If \(e\) is not the edge \(\{x, y\}\) then let
\(d_l(e) \defeq d_{l-1}(p, q) \cdot c_1(x, y)\).
The edge \(\{x, y\}\) takes the demand from \(G_{1}\) in addition,
therefore we define
\(d_l(x,y) \defeq d_{l-1}(u, v) \cdot c_1(x, y) + d_{1}(x,y)\).
\end{definition}

We recall here the following easy observation
that relates the treewidth of the supply graph of \(\mathcal{I}_1\)
to the treewidth of the supply graph of \(\mathcal{I}_l\).

\begin{lemma}[{\cite[Observation~4.4]{gupta2013sparsest}}]
\label{lem:tw-powering}
If the treewidth of the supply graph of \(\mathcal{I}_1\) is 
at most \(k\), then the treewidth of the supply graph of
\(\mathcal{I}_l\) is also at most \(k\). 
\end{lemma}

Corresponding to powering
instances, we can also recursively construct solutions to
\SparsestCut{N_{l}}
starting from a solution \(s_1\)
of \SparsestCut{N_{1}}.

\begin{definition}[Powering solutions separating \(u\) and \(v\)]
\label{def:solutions-powering}
 Given a base solution \(s_1\) of \SparsestCut{N_{1}} we
 construct a solution \(s_l\)  for \SparsestCut{N_{l}}
 recursively as follows.
 The solution \(s_{l}\) coincides with \(s_{1}\) on
 the vertices of \(G_{1}\).
 On the copy of \(G_{l-1}\) for an edge \(\{x,y\}\)
 of \(G_{1}\) we define \(s_{l}\) as follows.
 If \(s_{1}(x) = s_{1}(y)\) then let
 \(s_{l}((\{x,y\}, z)) \coloneqq s_{1}(x) = s_{1}(y)\)
 for all vertex \(z\) of \(G_{l-1}\),
 so as \(s_{l}\) cuts no edges in the copy of \(G_{l-1}\).
 If \(s_{1}(x) \neq s_{1}(y)\) then
 we define \(s_{l}\) so that the edges it cuts in the \(\{x,y\}\)-copy
 of \(G_{l-1}\) are exactly the copies of edges cut by \(s_{l-1}\)
 in \(G_{l-1}\).
 More precisely, let \((\{x,y\}, u)\) be identified with \(x\)
 and \((\{x,y\}, v)\) with \(y\).
 If \(s_{1}(x) = s_{l-1}(u)\)
 then we let \(s_{l}(\{x,y\}, z) \coloneqq s_{l-1}(z)\),
 otherwise we let \(s_{l}(\{x,y\}, z) \coloneqq - s_{l-1}(z)\).
\end{definition}

We now define the actual reduction. 
We construct a sequence of instances \(\{\mathcal{I}^*_1,
\mathcal{I}^*_2, \dots, \mathcal{I}^*_l\}\) 
where \(\mathcal{I}^*_l\) is obtained as in 
Definition~\ref{def:instance-powering} by applying the powering
operation to the base instance \(\mathcal{I}^*_1 = \mathcal{I}^*\)
and where \(N_1 \defeq n + 2, 
c_1 \defeq c\) and
\(d_1 \defeq d\). Note that by Lemma~\ref{lem:tw-powering},
the treewidth of the supply graph of \(\mathcal{I}^*_l\) is 
at most \(2\). We also construct a sequence of solutions
\(\{s^*_1, \dots, s^*_l\}\) where \(s^*_l\)
is obtained as in Definition~\ref{def:solutions-powering} by applying
the powering operation to the base solution \(s^*_1 = s^*\).
The final reduction maps
the instance \(\mathcal{I}\) and solution \(s\) of \MaxCUT{n}
to the instance \(\mathcal{I}^*_l\) and solution
\(s^*_l\) of \SparsestCut{N_l}{2} respectively. 
Completeness and soundness follows from 
\cite{gupta2013sparsest}. 

\begin{lemma}[{Completeness, \cite[Claim 4.2]{gupta2013sparsest}}]
\label{lem:sparsest-cut-completeness}
Let \(\mathcal{I}\) be an instance and \(s\) be a solution
of \MaxCUT{n},
and let their image be
the instance \(\mathcal{I}^*_l\) and solution \(s^*_l\)
of \SparsestCut{N_l}{2}{}, respectively.
Then the following holds
\begin{align*}
\val^n_{\mathcal{I}^*_l}(s^*_l) &= 1, &
\val^d_{\mathcal{I}^*_l}(s^*_l) &= l\val_{\mathcal{I}}(s).
\end{align*}
\end{lemma}

\begin{lemma}[{Soundness,
    \cite[Lemmas 4.3 and 4.7]{gupta2013sparsest}}]
\label{lem:sparsest-cut-soudness}
Let \(\mathcal{I}\) be an instance of \MaxCUT{n} 
and let 
\(\mathcal{I}^*_l\) be the instance of 
\SparsestCut{N_l}{2}{}
 it is mapped to. 
Then the instance \(\mathcal{I}^*_l\) has the following lower bound on
its optimum (the number of edges of \(\mathcal{I}\)
scales the \MaxCUT value between \(0\) and \(1\)).
\[
\OPT{\mathcal{I}^*_l} \geq
\frac{1}{1 + (l-1) \OPT{\mathcal{I}} / \size{E(\mathcal{I})}}.
\]  
\end{lemma}

Using the reduction framework of 
Section \ref{sec:reduct-fractional}
we now prove the main theorem of this section
about the LP and SDP
inapproximability of \SparsestCut.

\begin{proof}[Proof of Theorem~\ref{thm:hardness-sc}]
This is a simple application of
Lemmas \ref{lem:sparsest-cut-completeness} and
\ref{lem:sparsest-cut-soudness}
using Theorem~\ref{thm:red-fraction}
with matrices
\(M_1^{(n)}(\mathcal{I}_1, s_1) \defeq C_1(\mathcal{I}_1)\),
\(M_2^{(n)}(\mathcal{I}_1, s_1) \defeq 0\),
\(M_1^{(d)}(\mathcal{I}_1, s_1) \defeq 0\),
\(M_2^{(d)}(\mathcal{I}_1, s_1) \defeq 1\).
Hardness of the base problem \MaxCUT is provided by
Theorems~\ref{thm:LP-MaxCUT} and \ref{thm:SDP-MaxCUT},
and leads to
\(\eta_{\textrm{LP}} = \frac{5\varepsilon}{3 - \varepsilon}\)
and
\(\eta_{\textrm{SDP}} = \frac{3\varepsilon}{1 - 4\varepsilon}\).
\end{proof}

\subsection{\Problem{BalancedSeparator} with bounded-treewidth demand graph}
\label{subsec:balanced-separator}

In this section we show that the 
\Problem{BalancedSeparator} problem cannot be
approximated within any constant factor with small LPs even
when the demand graph has constant treewidth:

\begin{theorem}[LP-hardness for \Problem{BalancedSeparator}]
(Theorem~\ref{thm:balSepHardness} restated)
For any constant \(c_1 \ge 1\) there is another
constant \(c_2 \ge 1\) such that for all \(n\)
there is a demand function 
\(d \colon E(K_n) \rightarrow \R_{\ge 0}\) satisfying \(\tw([n]_d) \le c_2\)
so that \BalancedSeparator{n}{d} is LP-hard with an inapproximability
factor of \(c_1\).
\end{theorem}

We will reduce the \UniqueGames{n}{q} problem to the
\BalancedSeparator{2^qn}{d} problem for a fixed demand function \(d\)
to be defined below.
We reuse the reduction from \cite[Section~11.1]{khot2015unique}.
A bijection \(\pi \colon [q] \rightarrow [q]\)
acts on strings \(\cube{q}\) in
the natural way, i.e., \(\pi(x)_i \defeq x_{\pi(i)}\). For any
parameter \(p \in [0, 1]\), we denote by \(\VAR{x} \in_p \cube{q}\)
a random string where each coordinate \(\VAR{x}_i\) of 
\(\VAR{x}\) is \(-1\) with
probability \(p\) and \(1\) with probability \(1-p\). For a string \(x
\in \cube{q}\) we define \(x_{+} \defeq \lvert\{i \mid x_i = 1\}\rvert\) and \(x_{-} \defeq \lvert\{i \mid x_i = -1\}\rvert\). For a pair
of strings \(x, y \in \cube{q}\) we denote by \(xy\) the string in
\(\cube{q}\) formed by the coordinate-wise product of \(x\) and \(y\),
i.e., \((xy)_i \defeq x_iy_i\) for \(i \in [n]\). We are now ready
to proceed with the reduction.

Given an instance \(\mathcal{I} = \mathcal{I}(w, \pi)\) of \UniqueGames{n}{q} we construct
the instance \(\mathcal{I}^*\) of \BalancedSeparator{2^qn}{d}.
Let \(\varepsilon\) be a parameter to be
chosen later. The vertex set \(V\) of 
\(\mathcal{I}^*\)
is defined as \(V \defeq \left\{(x, i) \mid i \in [n], x \in \cube{q}
 \right\}\) so that \(|V| = 2^qn\). Let \(W \defeq \sum_{\{i, j\} \in
 E(K_n)} w(i, j)\) denote the total weight of
 the \UniqueGames{n}{q} instance
 \(\mathcal{I}\).
 For every \(i, j \in [n]\) and \(x, y \in \cube{q}\)
 there is an undirected edge \(\{(x, i), (y, j)\}\)
 in \(\mathcal{I}^*\) of capacity
 \(c((x, i), (y, j))\) which is defined as
\begin{align*}
 c\left((x, i), (y, j)\right)
 &\defeq \frac{w(i, j)\varepsilon^{\left(\pi_{i, j}(x)y\right)_{-}
 }(1-\varepsilon)^{\left(\pi_{i, j}(x)y\right)_{+}}}{2^qW}.
\intertext{
The demand function \(d\left((x, i), (y, j)\right)\)
 is defined
 for an unordered pair of vertices \(\{(x, i), (y, j)\}\)
 as
}
 d\left((x, i), (y, j)\right) &\defeq \begin{cases}
 \frac{1}{2^{2q-1}n}& \text{if } i = j \\
                                         0 & \text{otherwise}
                                        \end{cases}
\end{align*}
so that the total demand \(D\) is \(1\). Note that the demand graph
\([2^qn]_d\) is a disjoint of union of cliques of size \(2^q\) and
so \(\tw([2^qn]_d) = 2^q - 1 = O(1)\).
Given a solution \(s\) of \UniqueGames{n}{q} we map it to the solution \(s^*\) of \BalancedSeparator{2^qn}{d}
defined as \(s^* \defeq \left\{(x, i) \mid x_{s(i)} = 1 \right\}\). Note that the total demand cut by \(s^*\)
is \(\frac{1}{2} = \frac{D}{2} > \frac{D}{4}\)
since for every solution \(s(i) \in [q]\) there are exactly \(2^{q-1}\) strings in \(\cube{q}\) that have their
\(s(i)\)\textsuperscript{th} bit set to \(1\)
and \(2^{q-1}\) strings have their
\(s(i)\)\textsuperscript{th} bit set to \(-1\).
Thus \(s^*\) is a valid solution to the
\BalancedSeparator{2^qn}{d} problem and moreover is independent of
the instance \(\mathcal{I}^*\). We are now ready to show that this
reduction satisfies completeness.

\begin{lemma}[Completeness]\label{lem:bs-completeness}
Let \(\mathcal{I}\) and \(s\) be an instance
and a solution respectively of \UniqueGames{n}{q}.
Let \(\mathcal{I}^*\) and \(s^*\) be the instance and solution of
\BalancedSeparator{2^qn}{d} obtained from the reduction. Then

\begin{align*}
\frac{1}{2} - \val_{\mathcal{I}^*}(s^*) =
\left(\frac{1}{2} - \varepsilon\right) \val_{\mathcal{I}}(s)
\end{align*}
\end{lemma}

\begin{proof}
Let us sample a random edge
\((\VAR{i}, \VAR{j})\) from
the \UniqueGames{n}{q} instance \(\mathcal{I}\)
with probabilities proportional to \(w(i, j)\)
(i.e., \(\probability{\VAR{i} = i, \VAR{j} = j} = w(i, j) / W\)),
and independently sample \(\VAR{x} \in_{1/2} \cube{q}\)
and \(\VAR{z} \in_{\varepsilon} \cube{q}\).
Let \(\VAR{y} \coloneqq \pi_{\VAR{i}, \VAR{j}}(\VAR{x}) \VAR{z}\).

The claim follows by computing the probability of
\(\VAR{x}_{s(\VAR{i})} = \VAR{y}_{s(\VAR{j})}\)
in two different ways.

On the one hand, for a fixed edge  \((\VAR{i}\), \(\VAR{j})\)
of \(\mathcal{I}\),
depending on whether the edge is correctly labelled,
we have 

\begin{align}
  \probability[\VAR{i} = i, \VAR{j} = j,
  s(i) = \pi_{i,j}(s(j))]
  {\VAR{x}_{s(\VAR{i})} = \VAR{y}_{s(\VAR{j})}}
  &
  =
  \probability[\VAR{i} = i, \VAR{j} = j,
  s(i) = \pi_{i,j}(s(j))]
  {\VAR{z}_{s(j)} = 1}
  =
  1 - \varepsilon,
  \\
  \probability[\VAR{i} = i, \VAR{j} = j,
  s(i) \neq \pi_{i,j}(s(j))]
  {\VAR{x}_{s(\VAR{i})} = \VAR{y}_{s(\VAR{j})}}
  &
  =
  \probability[\VAR{i} = i, \VAR{j} = j,
  s(i) \neq \pi_{i,j}(s(j))]
  {\VAR{x}_{s(i)} = \VAR{y}_{s(j)}}
  =
  \frac{1}{2}
  .
\end{align}
Note that in the latter case \(\VAR{x}_{s(i)}\) and \(\VAR{y}_{s(j)}\)
are independent uniform binary variables.
Hence
\begin{align}
  \probability{s(\VAR{i}) = \pi_{\VAR{i},\VAR{j}}(s(\VAR{j})),
    \VAR{x}_{s(\VAR{i})} = \VAR{y}_{s(\VAR{j})}}
  &
  =
  (1 - \varepsilon) \val_{\mathcal{I}}(s),
  \\
  \probability{s(\VAR{i}) \neq \pi_{\VAR{i},\VAR{j}}(s(\VAR{j})),
    \VAR{x}_{s(\VAR{i})} = \VAR{y}_{s(\VAR{j})}}
  &
  =
  \probability[\VAR{i} = i, \VAR{j} = j,
  s(i) \neq \pi_{i,j}(s(j))]
  {\VAR{x}_{s(i)} = \VAR{y}_{s(j)}}
  \\
  &
  =
  \frac{1 - \val_{\mathcal{I}}(s)}{2}
  ,
  \intertext{leading to}
  \label{eq:bs-complete-UG}
  \probability{\VAR{x}_{s(\VAR{i})} = \VAR{y}_{s(\VAR{j})}}
  &
  =
  \frac{1}{2}
  +
  \left(
    \frac{1}{2} - \varepsilon
  \right)
  \val_{\mathcal{I}}(s)
  .
\end{align}

On the other hand, note that
\(\left((\VAR{x}, \VAR{i}), (\VAR{y}, \VAR{j})\right)\)
is a random edge from
\(\mathcal{I}^*\) with distribution
given by the weights \(c\left((x, i), (y, j)\right)\),
(i.e.,
\(\probability{\VAR{x} = x, \VAR{i} = i, \VAR{y} = y, \VAR{j} = j}
= c((x, i), (y, j))\)).
Recall that the cut \(s^*\) cuts an edge \(((x, i), (y, j))\)
if and only if \(x_{s(i)} \neq y_{s(j)}\).
It follows that
\begin{equation}
  \label{eq:bs-complete-reduced}
  \probability{\VAR{x}_{s(\VAR{i})} = \VAR{y}_{s(\VAR{j})}}
  =
  1 - \val_{\mathcal{I}^*}(s^*)
  .
\end{equation}
The claim now follows from Eqs.~\eqref{eq:bs-complete-UG}
and \eqref{eq:bs-complete-reduced}.
\end{proof}

Soundness of the reduction from \UniqueGames
to \Problem{BalancedSeparator}
is a reformulation of
\cite[Theorem 11.2]{khot2015unique}
without PCP verifiers:
\begin{lemma}[Soundness]\label{lem:bs-soundness}
  (Theorem 11.2 \cite{khot2015unique})
For every \(t \in \left(\frac{1}{2}, 1\right)\) there exists a
   constant
   \(b_t > 0\) such that the following holds. Let \(\varepsilon > 0\)
    be sufficiently small and let \(\mathcal{I} = \mathcal{I}(w, \pi)\)
     be an instance of \UniqueGames{n}{q} and let \(\mathcal{I}^*\)
	be the instance of \BalancedSeparator{2^qn}{d} as defined in
	Section \ref{subsec:balanced-separator}. If \(\OPT{\mathcal{I}}
      < 2^{-O(1/\varepsilon^2)}\) then \(\OPT{\mathcal{I}^*} > b_t
      \varepsilon^t\).
\end{lemma}

We are now ready to prove the main theorem of this section: that no
polynomial sized linear program can approximate the 
\Problem{BalancedSeparator} problem up to a constant factor.

\begin{theorem}\label{thm:hardness-bs}
For every \(q \ge 2, \delta > 0, t\in \left(\frac{1}{2}, 1\right)\)
 and \(k \ge 1\) there exists a
constant \(c > 0\) and a demand function \(d \colon E(K_n) \rightarrow
\R_{\ge 0}\) for every large enough \(n\), such that \(\tw([n]_d) =
2^q - 1\) and 
\begin{align*}
\fc\left(\BalancedSeparator{2^qn}{d},
 \delta + (\log{q})^{-1/2}, (\log{q})^{-t/2} \right)
 \ge cn^k.
\end{align*}
\end{theorem}

\begin{proof}
This statement follows immediately with
Lemmas~\ref{lem:bs-completeness} and \ref{lem:bs-soundness},
together with
Theorem~\ref{thm:red-simple} and Theorem~\ref{thm:hardnessofug}
with \(C_1 = 1 - \delta\),
\(S_1 = \frac{1}{q} + \delta\),
\(C_2 = \delta + (\log{q})^{-1/2}\)
and \(S_2 = \left(\log{q}\right)^{-t/2}\).
Note that the matrices as in Theorem~\ref{thm:red-simple} are chosen
as
\begin{align*}
M_1(\mathcal{I}, s) &= \frac{2}{1-2\varepsilon}, &
M_2(\mathcal{I}, s) = \frac{1 + \varepsilon}
{1 - 2 \varepsilon} \delta
\end{align*}
with \(\varepsilon = \left(\log{q}\right)^{-1/2}\). Since
\(M_1\) and \(M_2\) are constant nonnegative matrices
\(\LPrk{M_1} = \LPrk{M_2} = 1\).
\end{proof}

Finally, we can prove Theorem \ref{thm:balSepHardness} via choosing
the right parameters in Theorem~\ref{thm:hardness-bs}. 

\begin{proof}[Proof of Theorem \ref{thm:balSepHardness}]
Straightforward from Theorem~\ref{thm:hardness-bs}
 by choosing \(t = \frac{3}{4}, \delta = \left(\log{q}\right)^{-1/2}\) and \(q = 2^{{(2c_1)}^8}\)
so that the treewidth of the demand graph is bounded by 
\(c_2 = 2^q - 1 = 2^{2^{{(2c_1)}^8}} - 1\).
\end{proof}

\section{SDP hardness of \MaxCUT}
\label{sec:sdp-hardness-maxcut}

We now show that \MaxCUT cannot be approximated via small SDPs
within a factor of \(15/16 + \varepsilon\).
As approximation guarantees for an instance graph \(H\),
we shall use
\(C(H) = \alpha \size{E(H)}\) and
\(S(H) = \beta \size{E(H)}\)
for some constants \(\alpha\) and \(\beta\),
and for brevity we will only write \(\alpha\) and \(\beta\).
\begin{theorem}
  \label{thm:SDP-MaxCUT}
  For any \(\delta, \varepsilon > 0\)
  there are infinitely many \(n\)
  such that there is a graph \(G\) with \(n\) vertices and
  \begin{equation}
    \label{eq:SDP-MaxCUT}
    \fcSDP\left(\MaxCUT{G}, \frac{4}{5} - \varepsilon, \frac{3}{4} + \delta\right)
    = n^{\Omega(\log n / \log \log n)}
    .
  \end{equation}
\begin{proof}
Recall \cite[Theorem~4.5]{schoenebeck2008linear}
applied to the predicate \(P = (x_{1} + x_{2} + x_{3} = 0)\) (mod \(2\)):
For any \(\gamma, \delta > 0\), and large enough \(m\),
there is an instance \(\mathcal{I}\) of \MaxXOR[0]{3}
on \(m\) variables
with \(\OPT{\mathcal{I}} \leq 1/2 + \delta\)
but having a Lasserre solution after \(\Omega(m^{1 - \gamma})\) rounds
satisfying all the clauses.
By \cite[Theorem~6.4]{lee2014lower}, we obtain that
for any \(\delta, \varepsilon > 0\) for infinitely many \(m\)
\begin{equation}
  \label{eq:Max-3-XOR-LP_hard}
  \fcSDP\left(\MaxXOR[0]{3}, 1 - \varepsilon, \frac{1}{2} + \delta\right)
  = m^{\Omega(\log m / \log \log m)}
  .
\end{equation}

We reuse the reduction from \MaxXOR[0]{3} to \MaxCUT in
\cite[Lemma~4.2]{trevisan2000gadgets}.
Let \(x_{1}\), \dots, \(x_{m}\) be the variables for \MaxXOR[0]{3}.
For every possible clause \(C = (x_{i} + x_{j} + x_{k} = 0)\),
we shall use the gadget graph \(H_{C}\) from
\cite[Figure~4.1]{trevisan2000gadgets},
reproduced in Figure~\ref{fig:MaxXOR-to-MaxCUT}.
We shall use the graph \(G\), which is the union of all the
gadgets \(H(C)\) for all possible clauses.
The vertices \(0\) and \(x_{1}\), \dots, \(x_{m}\) are shared by the
gadgets, the other vertices are unique to each gadget.
\begin{figure}[ht]
  \small
  \centering
  \includegraphics{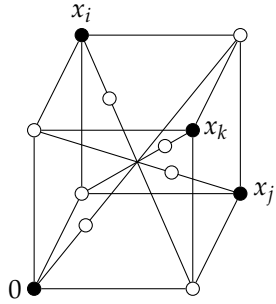}
  \caption{\label{fig:MaxXOR-to-MaxCUT}%
    The gadget \(H_{C}\)
    for the clause \(C = (x_{i} + x_{j} + x_{k} = 0)\)
    in the reduction from \MaxXOR[0]{3} to \MaxCUT.
    Solid vertices are shared by gadgets, the empty ones are local
    to the gadget.}
\end{figure}

A \MaxXOR[0]{3} instance \(\mathcal{I} = \{C_{1}, \dots, C_{l}\}\)
is mapped to
the union \(G_{\mathcal{I}} = \bigcup_{i} H(C_{i})\)
of the gadgets of the clauses \(C_{i}\) in \(\mathcal{I}\),
which is an induced subgraph of \(G\).

A feasible solution, i.e., an assignment
\(s \colon \{x_{1}, \dots, x_{m}\} \to \{0, 1\}\)
is mapped to a vertex set \(s^{*}\) satisfying
the following conditions:
\begin{enumerate*}
\item \(x_{i} \in s^{*}\) if and only if \(s(x_{i}) = 1\)
\item \(0 \notin s^{*}\)
\item on every gadget \(H(C)\)
  the set \(s^{*}\) cuts the maximal number of edges subject to
  the previous two conditions
\end{enumerate*}
It is easy to see that \(s^{*}\) cuts
\(16\) out of the \(20\) edges of every \(H(C)\)
if \(s\) satisfies \(C\),
and it cuts \(14\) edges if \(s\) does not satisfy \(C\).
Therefore
\begin{equation}
  \val^{\MaxCUT{G}}_{G_{\mathcal{I}}}(s^{*})
  =
  \frac{14 + 2 \val^{\MaxXOR[0]{3}}_{\mathcal{I}}(s)}{20}
  ,
\end{equation}
which by rearranging provides the completeness of the reduction:
\begin{equation}
  1 - \varepsilon - \val^{\MaxXOR[0]{3}}_{\mathcal{I}}(s)
  =
  10 \left[
    \frac{4}{5}
    -  \frac{\varepsilon}{10}
    - \val^{\MaxCUT{G}}_{G_{\mathcal{I}}}(s^{*})
  \right]
  .
\end{equation}
It also follows from the construction
that \(\val^{\MaxCUT{G}}_{G_{\mathcal{I}}}\) achieves its maximum on
a vertex set of the form \(s^{*}\):
given a vertex set \(X\) of \(G\),
if \(0 \notin X\) then let
let \(s(x_{i}) = 1\) if \(x_{i} \in X\),
and \(s(x_{i}) = 0\) otherwise.
If \(x_{i} \in X\) then we do it the other way around:
\(s(x_{i}) = 1\) if and only if \(x_{i} \notin X\).
This definition makes \(s^{*}\)
on the vertices
\(0\), \(x_{1}\), \dots, \(x_{m}\)
either agree with \(X\)
(if \(0 \notin X\))
or to be complement of \(X\)
(if \(0 \in X\)).
Then \(\val^{\MaxCUT{G}}_{G_{\mathcal{I}}}(s^{*}) \geq
\val^{\MaxCUT{G}}_{G_{\mathcal{I}}}(X)\)
by construction.
This means
\begin{equation}
  \max \val^{\MaxCUT{G}}_{G_{\mathcal{I}}}
  =
  \frac{14 + 2 \max \val^{\MaxXOR[0]{3}}_{\mathcal{I}}}{20}
  .
\end{equation}
Thus if \(\max \val^{\MaxXOR[0]{3}}_{\mathcal{I}} \leq 1/2 + \delta\)
then
\(\max \val^{\MaxCUT{G}}_{G_{\mathcal{I}}} \leq 3/4 + \delta / 10\).
Therefore we obtain a reduction with guarantees
\(C_{\MaxCUT{G}} = 4/5 - \varepsilon / 10\),
\(S_{\MaxCUT{G}} = 3/4 + \delta / 10\),
\(C_{\MaxXOR[0]{3}} = 1 - \varepsilon\),
\(S_{\MaxXOR[0]{3}} = 1/2 + \delta\),
proving
\begin{equation}
 \begin{split}
  \fcSDP\left(\MaxCUT{G}, \frac{4}{5} - \frac{\varepsilon}{10}, \frac{3}{4} + \frac{\delta}{10}\right)
  &
  \geq
  \fcSDP\left(\MaxXOR[0]{3}, 1 - \varepsilon, \frac{1}{2} + \delta\right)
  \\
  &
  = m^{\Omega(\log m / \log \log m)}
  = n^{\Omega(\log n / \log \log n)}
  ,
 \end{split}
\end{equation}
where \(n = O(m^{3})\) is the number of vertices of \(G\).
\end{proof}
\end{theorem}

\section{Lasserre relaxation is suboptimal for \IndependentSet{G}}
\label{sec:Lasserre-independent-set}

Applying reductions within Lasserre hierarchy formulations, we will
now derive a new lower bound on the Lasserre integrality gap for the
\IndependentSet problem, establishing that the Lasserre hierarchy is
suboptimal: there exists a linear-sized LP formulation for the
\IndependentSet problem with approximation guarantee \(2\sqrt{n}\),
whereas there exists a family of graphs with Lasserre integrality gap
\(n^{1-\gamma}\)
after \(\Omega(n^\gamma)\)
rounds for arbitrary small \(\gamma\).
While this is expected assuming P vs.~NP, our result is unconditional.
It also complements previous integrality gaps, like
\(n/2^{O(\sqrt{\log n \log \log n})}\)
for \(2^{\Theta(\sqrt{\log n \log \log n})}\)
rounds in \cite{tulsiani2009csp}, and others in
\cite{au2011complexity}, e.g., \(\Theta(\sqrt{n})\)
rounds of Lasserre are required for deriving the exact optimum.

For \IndependentSet{G}, the base functions of the Lasserre hierarchy
are the indicator functions \(Y_{v}\)
that a vertex \(v\) is contained in a feasible solution
(which is an independent set),
i.e., \(Y_{v}(I) \coloneqq \chi(v \in I)\).

\begin{theorem}
  \label{thm:independent-set_Lasserre}
  For any small enough \(\gamma > 0\)
  there are infinitely many \(n\),
  such that there is a graph \(G\) with \(n\) vertices
  with the largest independent set of \(G\)
  having size \(\alpha(G) = O(n^{\gamma})\)
  but there is a \(\Omega(n^{\gamma})\)-round Lasserre solution
  of size \(\Theta(n)\), i.e., the integrality gap is
  \(n^{1-\gamma}\). However \(\fc(\IndependentSet{G}, 2\sqrt{n}) \leq 3n + 1\).
\begin{proof}
The statement \(\fc(\IndependentSet{G}, 2\sqrt{n}) \leq 3n + 1\) is
\cite[Lemma~5.2]{bazzi2015no}. For the integrality gap construction,
we apply Theorem~\ref{thm:hardness-k-CSP} with the following choice of
parameters.  We shall use \(N\) for the number of variables,
as \(n\) will be the number of vertices of \(G\).
The parameters \(q\) and \(\varepsilon\) are fixed to arbitrary
values.
The parameter \(\kappa\) is chosen close to \(1\),
and \(\delta\) is chosen to be a large constant;
the exact values will be determined later.
The number of variables \(N\) will vary, but will be large enough
depending on the parameters already chosen.
The parameters \(\beta\) and \(k\) are chosen so that the required
lower and upper bounds on \(\beta\) are approximately the same:
\begin{align}
  k
  &
 \begin{aligned}[t]
   &
  \coloneqq
  \left\lfloor
    \frac{(1 - \kappa) (\delta - 1) \log N
      - \Theta(\delta \log \log N)}
    {\log q}
  \right\rfloor
  \\
  &
  =
  \frac{(1 - \kappa) (\delta - 1) \log N
    - \Theta(\delta \log \log N)}
  {\log q}
  =
  \Theta(\log N)
 \end{aligned}
  \\
  \beta
  &
  \coloneqq
  \frac{1}{N} \left\lceil
    \frac{6 N q^{k} \ln q}{\varepsilon^{2}}
  \right\rceil
  =
  q^{k + o(1)}
  =
  N^{(1 - \kappa) (\delta - 1) - \Theta(\delta \log \log N / \log N)}
  .
\end{align}
Thus \(\beta \geq (6 q^{k} \ln q) / \varepsilon^{2}\),
and for large enough \(N\),
we also have
\[\beta \leq N^{(1 - \kappa) (\delta - 1)} /
(10^{8 (\delta - 1)} k^{2 \delta + 0.75}).\]
(The role of the term \(\Theta(\delta \log \log N)\) in \(k\)
is ensuring this upper bound.
Rounding ensures that \(k\) and \(\beta N\) are integers.)
By the theorem,
there is a \(k\)-CSP \(\mathfrak{I}\) on \(N\) variables
\(x_{1}\), \dots, \(x_{N}\)
and clauses \(C_{1}, \dotsc, C_{m}\) coming from a predicate \(P\)
such that \(\OPT{\mathcal{I}} = O((1 + \varepsilon) / q^{k})\)
and there is a pseudoexpectation \(\pseudoexpectation_{\mathcal{I}}\)
of degree at least \(\eta N / 16\)
with \(\pseudoexpectation_{\mathcal{I}}(\val_{\mathcal{I}}) = 1\).
Here
\begin{align}
  m
  &
  \coloneqq
  \beta N
  =
  N^{(1 - \kappa \pm o(1)) (\delta - 1)}
  ,
  \\
  \eta N / 16
  &
  =
  N^{\kappa \pm o(1)}
  .
\end{align}

Let \(a\) denote the number of satisfying partial assignments of
\(P\).
A uniformly random assignment satisfies an \(a / q^{k}\) fraction of
the clauses in expectation, therefore
\(a / q^{k} \leq \OPT{\mathcal{I}} = O((1 + \varepsilon) / q^{k})\),
i.e., \(a = \Theta(1 + \varepsilon)\).

Let \(G\) be the conflict graph of \(\mathcal{I}\),
i.e., the vertices of \(G\) are pairs \((i, s)\)
with \(i \in [m]\) and \(s\) a satisfying partial assignment \(s\)
of clause \(C_{i}\)
with domain the set of free variables of \(C_{i}\).
Two pairs \((i, s)\) and \((j, t)\) assignments are
adjacent as vertices of \(G\)
if and only if the partial assignments \(s\) and \(t\) conflict,
i.e., \(s(x_{j}) \neq t(x_{j})\) for some variable \(x_{j}\)
on which both \(s\) and \(t\) are defined.
Thus \(G\) has
\begin{equation}
  n \coloneqq a m = N^{(1 - \kappa) (\delta - 1) \pm o(1)}
\end{equation}
vertices.

Given an assignment \(t \colon \{x_{1}, \dotsc, x_{N}\} \to [q]\)
we define the independent set \(t^{*}\) of \(G\) as the set of
partial assignments \(s\) compatible with \(t\).
(Obviously, \(t^{*}\) is really an independent set.)
This provides a mapping \(*\) from the set of assignments
of the \(x_{1}\), \dots, \(x_{N}\) to
the set of independent set of \(G\).
Clearly, \(\val_{G}(t^{*}) = m \val_{\mathcal{I}}(t)\),
as \(t^{*}\) contains one vertex per clause satisfied by \(t\).
It is easy to see that every independent set \(I\) of \(G\)
is a subset of some \(t^{*}\), and hence
\begin{equation}
  \label{eq:independent-set-Lasserre-sound}
  \OPT{G} = m \OPT{\mathcal{I}} = m O((1 + \varepsilon) / q^{k})
  =
  O(N)
  =
  O(n^{1 / [(1 - \kappa) (\delta - 1) \pm o(1)]})
  .
\end{equation}

We define a pseudoexpectation \(\pseudoexpectation_{G}\)
of degree \(\eta N / 16 k \)
for \(G\) as a composition of \(*\) and the pseudoexpectation
\(\pseudoexpectation_{\mathcal{I}}\) of
the CSP instance \(\mathcal{I}\):
\begin{equation}
  \label{eq:independent-set-Lasserre-reduction}
  \pseudoexpectation_{G}(F) \coloneqq
  \pseudoexpectation_{\mathcal{I}}(F \circ *)
  .
\end{equation}
Recall that \(X_{x_{j} = b}\) is the indicator that \(b\)
is assigned to the variable \(x_{j}\),
and \(Y_{(i, s)}\) is the indicator that \((i, s)\) is part of the
independent set.
Note that for \(s \in V(G)\), we have
\(Y_{(i, s)} \circ * = \prod_{x_{j} \in \dom s} X_{x_{j} = s(x_{j})}\)
is of degree at most \(k\), and therefore
\(\deg(F \circ *) \leq k \deg F\),
showing that \(\pseudoexpectation_{G}\) is well-defined.
Clearly \(\pseudoexpectation_{G}\) is a pseudo-expectation,
as so is \(\pseudoexpectation_{\mathcal{I}}\).

Now, letting \(s \sim C_{i}\) denote that \(s\)
is a satisfying partial assignment for \(C_{i}\):
\begin{equation}
  \val_{G} \circ *
  =
  \sum_{(i, s) \in V(G)} Y_{(i, s)} \circ *
  =
  \sum_{i \in [m]}
  \sum_{s \sim C_{i}} \prod_{x_{j} \in \dom s} X_{x_{j} = s(x_{j})}
  =
  \sum_{i \in [m]} C_{i} = m \val_{\mathcal{I}}
  ,
\end{equation}
and hence
\begin{equation}
  \label{eq:independent-set-Lasserre-complete}
  \pseudoexpectation_{G}(\val_{G}) =
  m \cdot
  \pseudoexpectation_{\mathcal{I}}(\val_{\mathcal{I}})
  = m
  = n / a = \Theta(n)
  ,
\end{equation}
showing \(\SOS_{\eta N / 16 k}(G) \geq m\).
The number of rounds is
\begin{equation}
  \label{eq:independent-set-Lasserre-rounds}
  \eta N / 16 k =
  n^{[\kappa \pm o(1)]/[(1 - \kappa) (\delta - 1) \pm o(1)]}
  .
\end{equation}
From Equations~\eqref{eq:independent-set-Lasserre-sound},
\eqref{eq:independent-set-Lasserre-complete} and
\eqref{eq:independent-set-Lasserre-rounds}
the theorem follows with an appropriate choice of \(\kappa\) and
\(\delta\) depending on \(\gamma\).
\end{proof}
\end{theorem}

\section{From Sherali–Adams reductions to general LP reductions}
\label{sec:extend-Sherali-Adams}

There are several reductions between Sherali–Adams solutions
of problems in the literature.
Most of these reductions do not make essential use of
the Sherali–Adams hierarchy.
The reduction mechanism introduced in Section~\ref{sec:reductions}
allows us to directly execute them in
the linear programming framework.
As an example, we extend the Sherali–Adams reductions
from \UniqueGames to
various kinds of CSPs from \cite{bazzi2015no} to the general LP case.
These CSPs are used in \cite{bazzi2015no} as
intermediate problems for reducing to non-uniform \VertexCover and
\hyperVertexCover{Q},
hence composing the reductions here
with the ones in \cite{bazzi2015no} yield
direct reductions from \(\UniqueGames\) to \VertexCover and
\hyperVertexCover{Q}.

\subsection{Reducing \Problem{UniqueGames} to \OneFreeCSP{}}
\label{sec:reduce-UG-1FCSP}

We demonstrate the generalization to LP reductions by transforming the
Sherali–Adams reduction from  \UniqueGames to \OneFreeCSP{} in
\cite{bazzi2015no}.

\begin{definition}
  \label{def:1F-CSP}
  A \emph{one-free bit} CSP (\OneFreeCSP{} for short)
  is a CSP
  where every clause has exactly two satisfying assignments over its
  free variables.
\end{definition}

\begin{theorem}
  \label{thm:1F-CSP}
  With small numbers \(\eta, \varepsilon, \delta > 0\)
  positive integers \(t\), \(q\), \(\Delta\)
  as in \cite[Lemma~3.4]{bazzi2015no},
  we have for any \(0 < \zeta < 1\) and \(n\) large enough
  \begin{equation}
    \fc(\UniqueGames[\Delta]{n}{q}, 1 - \zeta, \delta)
    - n \Delta^{t} q^{t+1}
    \leq
    \fc(\OneFreeCSP, (1 - \varepsilon) (1 - \zeta t), \eta)
  \end{equation}
\begin{proof}
Let \(V = \{0, 1\} \times [n]\)
denote the common set of vertices
of all the instances of \UniqueGames[\Delta]{n}{q}.
The variables of \OneFreeCSP are chosen to be
all the \(\sprod{v}{z}\)
for \(v \in V\) and \(z \in \{-1, +1\}^{[q]}\).
(Here \(\sprod{v}{z}\) stands for the pair of \(v\) and \(z\).)
Given a \UniqueGames[\Delta]{n}{q} instance \((G, w, \pi)\),
we define an instance \((G, w, \pi)^{*}\) of \OneFreeCSP as follows.

Let \(v\) be any vertex of \(G\),
and let \(u_{1}\), \dots, \(u_{t}\) be
vertices adjacent to \(v\)
(allowing the same vertex to appear multiple times).
Furthermore, let \(x \in \{-1, +1\}^{[q]}\)
and let \(S\) be a subset of \([q]\) of size \((1 - \varepsilon) q\).
We introduce the clause \(C(v, u_{1}, \dotsc, u_{t}, x, S)\)
as follows, which is an approximate test for the edges
\(\{v, u_{1}\}\), \dots, \(\{v, u_{t}\}\) to be correctly labelled.
\begin{equation}
  \begin{aligned}
    C(v, u_{1}, \dotsc, u_{t}, x, S) &\coloneqq
    \exists b \in \{-1, +1\} \; \forall i \in [t]
    \; \forall z \in \{-1,+1\}^{[q]}
    \\
    &
    \begin{cases}
      \sprod{u_{i}}{z} = b
      &
      \text{if }
      \pi_{v, u_{i}}(z) \restriction S =
      x \restriction S
      ,
      \\
      \sprod{u_{i}}{z} = -b
      &
      \text{if }
      \pi_{v, u_{i}}(z) \restriction S =
      -x \restriction S
      .
    \end{cases}
  \end{aligned}
\end{equation}
We will define a probability distribution on clauses, and
the \emph{weight} of a clause will be its probability.

First we define a probability distribution \(\mu_{1}\)
on edges of \(G\)
proportional to the weights.
More precisely, we define a distribution on pairs of
adjacent vertices \((\VAR{v}, \VAR{u})\):
\begin{equation}
  \probability{\{\VAR{v}, \VAR{u}\} = \{v, u\}}
  \coloneqq \frac{w(v, u)}{\sum_{i,j} w(i,j)}
  ,
\end{equation}
therefore for the objective of \UniqueGames[\Delta]{n}{q} we obtain
\begin{equation}
  \label{eq:Unique-Games-expectation}
  \val^{\UniqueGames[\Delta]{n}{q}}_{(G, w, \pi)}(s) =
   \expectation{s(\VAR{v})
    = \pi_{\VAR{v}, \VAR{u}}(s(\VAR{u}))
  }
\end{equation}
Let \(\mu_{1}^{v}\) denote the marginal
of \(\VAR{v}\) in the distribution \(\mu_{1}\),
and \(\mu_{1}^{u \mid v}\) denote the conditional distribution
of \(\VAR{u}\) given \(\VAR{v} = v\).

Now we define a distribution \(\mu_{t}\)
on vertices \(\VAR{v}\), \(\VAR{u_{1}}\), \dots, \(\VAR{u_{t}}\)
such that \(\VAR{v}\) has the marginal distribution
\(\mu_{1}^{v}\), and given \(\VAR{v} = v\),
the vertices \(\VAR{u_{1}}\), \dots, \(\VAR{u_{t}}\)
are chosen mutually independently,
each with the conditional distribution \(\mu_{1}^{u \mid v}\).
Thereby every pair \((\VAR{v}, \VAR{u_{i}})\)
has marginal distribution \(\mu_{1}\).

Finally,
\(\VAR{x} \in \{-1, +1\}^{[q]}\)
and
\(\VAR{S} \subseteq [q]\)
are chosen randomly and independently of each other and the
vertices \(\VAR{v}\), \(\VAR{u_{1}}\), \dots, \(\VAR{u_{t}}\),
subject to the restriction \(\size{\VAR{S}} = (1 - \varepsilon) q\)
on the size of \(\VAR{S}\).
This finishes the definition of the distribution of clauses,
in particular,
\begin{equation}
  \label{eq:1F-CSP-expectation}
  \val^{\OneFreeCSP}_{(G, w, \pi)^{*}} (p)
  = \expectation{
    C(\VAR{v}, \VAR{u_{1}}, \dotsc, \VAR{u_{t}}, \VAR{x}, \VAR{S})
    [p]}
\end{equation}
for all evaluation \(p\).

Feasible solutions are translated via
\begin{equation}
  \label{eq:1F-CSP_solution}
  s^{*}(\sprod{v}{z}) \coloneqq z_{s(v)}.
\end{equation}

Soundness of the reduction,
i.e., \eqref{eq:red-simple-sound}
follows from \cite[Lemma~3.4]{bazzi2015no}.

Completeness, i.e., \eqref{eq:red-simple-complete},
easily follows from an extension of the argument
in \cite[Lemma~3.5]{bazzi2015no}.  The main estimation comes
from the fact that the clause \(C(v, u_{1}, \dotsc, u_{t}, x, S)\)
is satisfied if the edges
\(\{v, u_{1}\}\), \dots, \(\{v, u_{t}\}\)
are all correctly labeled and the label \(s(v)\) of \(v\)
lies in \(S\):
\begin{equation}
  C(v, u_{1}, \dotsc, u_{t}, x, S) [s^{*}] \geq
  \chi[s(v) = \pi_{v,u_{i}}(s(u_{i})), \forall i \in [t];
  s(v) \in S]
  .
\end{equation}

Let us fix the vertices \(v\), \(u_{1}\), \dots, \(u_{k}\)
and take expectation over \(x\) and \(S\):
\begin{equation}
 \begin{split}
  \expectation(\VAR{x}, \VAR{S}){
    C(v, u_{1}, \dotsc, u_{t}, \VAR{x}, \VAR{S}) [s^{*}]}
  &
  \geq
  (1 - \varepsilon) \chi[s(v) = \pi_{v,u_{i}}(s(u_{i})),
  \forall i \in [t]]
  \\
  &
  \geq
  (1 - \varepsilon)
  \left(
    \sum_{i \in [t]}
    \chi[s(v) = \pi_{v,u_{i}}(s(u_{i}))]
    - t + 1
  \right)
  .
 \end{split}
\end{equation}
We build a nonnegative matrix \(M\) out of the difference of the two
sides of the inequality.
The difference depends only partly on \(s\):
namely, only on the values of \(s\) on the vertices
\(v\), \(u_{1}\), \dots, \(u_{t}\).
Therefore we also build a smaller variant \(\widetilde{M}\) of \(M\)
making this dependence explicit,
which will be the key to establish low LP-rank later:
\begin{equation}
 \begin{split}
   &
  \widetilde{M}_{v, u_{1}, \dotsc, u_{t}}((G, w, \pi),
  s \restriction \{v, u_{1}, \dotsc, u_{t}\})
  =
  M_{v, u_{1}, \dotsc, u_{t}}((G, w, \pi), s)
  \\
  &
  \phantom{M}
  \coloneqq
  \expectation(\VAR{x}, \VAR{S})
  {C(v, u_{1}, \dotsc, u_{t}, \VAR{x}, \VAR{S}) [s^{*}]}
  -
  (1 - \varepsilon)
  \left(
    \sum_{i \in [t]}
    \chi[s(v) = \pi_{v,u_{i}}(s(u_{i}))]
    - t + 1
  \right)
  \\
  &
  \phantom{M}
  \geq
  0
  .
 \end{split}
\end{equation}

Taking expectation provides
\begin{equation}
 \begin{split}
  \val^{\OneFreeCSP}_{(G, w, \pi)^{*}} (s^{*})
  &
  =
  \expectation{C(\VAR{v}, \VAR{u_{1}}, \dotsc, \VAR{u_{t}},
    \VAR{x}, \VAR{S}) [s^{*}]}
  \\
  &
  =
  (1 - \varepsilon)
  \left(
    \sum_{i \in [t]}
    \probability{s(\VAR{v}) =
      \pi_{\VAR{v}, \VAR{u_{i}}}(s(\VAR{u_{i}}))}
    - t + 1
  \right)
  +
  \expectation{
    M_{\VAR{v}, \VAR{u_{1}}, \dotsc, \VAR{u_{t}}}((G, w, \pi), s)}
  \\
  &
  =
  (1 - \varepsilon)
  (t \val^{\UniqueGames{n}{q}}_{(G, w, \pi)} (s) - t + 1)
  + \expectation{%
    M_{\VAR{v}, \VAR{u_{1}}, \dotsc, \VAR{u_{t}}}((G, w, \pi), s)},
 \end{split}
\end{equation}
and hence after rearranging we obtain,
no matter what \(\zeta\) is
\begin{equation}
\label{eq:finalMap1F}
  1 - \zeta - \val^{\UniqueGames{n}{q}}_{(G, w, \pi)} (s)
  = \frac{(1 - \varepsilon) (1 - \zeta t)
    - \val^{\OneFreeCSP}_{(G, w, \pi)^{*}} (s^{*}) +
    \expectation{M_{\VAR{v}, \VAR{u_{1}}, \dotsc, \VAR{u_{t}}}((G, w, \pi), s)}}
  {t (1 - \varepsilon)}.
\end{equation}
(Note that Equation~\eqref{eq:finalMap1F} is not affine
due to the last term in the numerator.)

Here the last term in the numerator is the matrix \(M_{2}\)
in the reduction Definition~\ref{def:red-simple}
(up to the constant factor of the denominator).
We show that it has low LP rank:
\begin{multline}
  \expectation{M_{\VAR{v}, \VAR{u_{1}}, \dotsc, \VAR{u_{t}}}((G, w, \pi), s)}
  \\
  =
  \sum_{\substack{v, u_{1}, \dotsc, u_{t} \\
      f \colon \{v, u_{1}, \dotsc, u_{t} \} \to [q]}}
  \left(
    \probability{\VAR{v} = v, \VAR{u_{1}} = u_{1}, \dotsc,
      \VAR{u_{t}} = u_{t}}
    \widetilde{M}_{v, u_{1}, \dotsc, u_{t}}((G, w, \pi), f)
  \right)
  \\
  \cdot
  \chi(f = s \restriction \{v, u_{1}, \dotsc, u_{t}\}),
\end{multline}
 i.e.,  the expectation can be written as the sum of at most \(n \Delta^{t}
q^{t+1}\) nonnegative rank-\(1\) factors. Therefore
the claim follows from Theorem~\ref{thm:red-simple}.
\end{proof}
\end{theorem}

\subsection{Reducing \Problem{UniqueGames} to \NotEqualCSP{Q}}
\label{sec:not-equal-CSP}

\begin{definition}
  \label{def:not-equal-CSP}
  A \emph{not equal CSP} (\NotEqualCSP{Q} for short)
  is a CSP with value set \(\mathbb{Z}_{Q}\),
  the additive group of integers modulo \(Q\),
  where every clause has the form
  \(\bigland_{i=1}^{k} x_{i} \neq a_{i}\)
  for some constants \(a_{i}\).
\end{definition}

\begin{theorem}
  \label{thm:not-equal-CSP}
  With small numbers \(\eta, \varepsilon, \delta > 0\)
  positive integers \(t\), \(q\), \(\Delta\)
  as in \cite[Lemma~3.4]{bazzi2015no},
  we have for any \(0 < \zeta < 1\) and \(n\) large enough
  \begin{equation}
    \fc(\UniqueGames[\Delta]{n}{q}, 1 - \zeta, \delta)
    - n \Delta^{t} q^{t+1}
    \leq
    \fc(\NotEqualCSP{Q}, (1 - \varepsilon) (1 - 1/q)
    (1 - \zeta t), \eta)
  \end{equation}
\begin{proof}
  The proof is similar to that of Theorem~\ref{thm:1F-CSP}, with the
  value set \(\{-1, +1\}\)
  consistently replaced with \(\mathbb{Z}_{Q}\).
  Let \(V = \{0, 1\} \times [n]\)
  again denote the common set of vertices of all the instances of
  \UniqueGames{n}{q}.  The variables of \NotEqualCSP{Q} are chosen to
  be all the \(\sprod{v}{z}\)
  for \(v \in V\)
  and \(z \in \mathbb{Z}_{Q}^{[q]}\).

  To simplify the argument,
  we now introduce additional \emph{hard constraints}, i.e., which
  have to be satisfied by any assignment.
  This can be done without loss of generality as these
  hard constraints can be eliminated by using only one variable from
  every coset of \(\mathbb{Z}_{Q} \allOne\)
  and substituting out the other variables. The resulting CSP will be
  still a not equal CSP, however this would break the natural
  symmetry of the structure. Let
  \(\allOne \in \mathbb{Z}_{Q}^{[q]}\)
  denote the element with all coordinates \(1\).
  We introduce the \emph{hard constraints}
\begin{equation}
  \label{eq:not-equal-CSP_hard-constraints}
  \sprod{v}{z + \lambda \allOne} =
  \sprod{v}{z} + \lambda
  \qquad
  (\lambda \in \mathbb{Z}_{Q}).
\end{equation}

Given a \UniqueGames[\Delta]{n}{q} instance \((G, w, \pi)\),
we now define an instance \((G, w, \pi)^{*}\) of \NotEqualCSP{Q}
as follows.
Let \(v\) be any vertex of \(G\),
and let \(u_{1}\), \dots, \(u_{t}\) be
vertices adjacent to \(v\)
(allowing the same vertex to appear multiple times).
Furthermore, let \(x \in \mathbb{Z}_{Q}^{[q]}\)
and let \(S\) be a subset of \([q]\) of size \((1 - \varepsilon) q\).
We introduce the clause \(C(v, u_{1}, \dotsc, u_{t}, x, S)\)
as follows, which is once more an approximate test for the edges
\(\{v, u_{1}\}\), \dots, \(\{v, u_{t}\}\) to be correctly labeled.
\begin{equation}
  \begin{aligned}
    C(v, u_{1}, \dotsc, u_{t}, x, S) &\coloneqq
    \forall i \in [t]
    \forall z \in \mathbb{Z}_{Q}^{[q]}
    \\
    &
      \sprod{u_{i}}{z} \neq 0
      \qquad
      \text{if }
      \pi_{v, u_{i}}(z) \restriction S =
      x \restriction S
      .
  \end{aligned}
\end{equation}
The \emph{weight} of a clause is defined as its probability
using the same distribution
on vertices \(\VAR{v}\), \(\VAR{u_{1}}\), \dots, \(\VAR{u_{t}}\)
as in Theorem~\ref{thm:1F-CSP},
and randomly and independently chosen
\(\VAR{x} \in \mathbb{Z}_{Q}^{[q]}\)
and \(S \subseteq [q]\) with \(\size{S} = \varepsilon [q]\).
This is the analogue of the distribution
in Theorem~\ref{thm:1F-CSP},
in particular,
\begin{align}
  \label{eq:not-equal-CSP_Unique-Games-expectation}
  \val^{\UniqueGames[\Delta]{n}{q}}_{(G, w, \pi)}(s)
  &
  = \expectation{s(\VAR{v})
    = \pi_{\VAR{v}, \VAR{u_{i}}}(s(\VAR{u_{i}}))
  }
  &
  (i &\in [t])
  ,
  \\
  \val^{\OneFreeCSP}_{(G, w, \pi)^{*}} [p]
  &
  = \expectation{
    C(\VAR{v}, \VAR{u_{1}}, \dotsc, \VAR{u_{t}}, \VAR{x}, \VAR{S})
    [p]}
  .
\end{align}

Feasible solutions are translated via
\begin{equation}
  \label{eq:not-equal-CSP_solution}
  s^{*}(\sprod{v}{z}) \coloneqq z_{s(v)}
  ,
\end{equation}
which clearly satisfy the hard constraints
\eqref{eq:not-equal-CSP_hard-constraints}.

The reduction is sound by \cite[Lemma~6.9]{bazzi2015no}.
For completeness, we follow a similar approach to
\cite[Lemma~6.10]{bazzi2015no} and of Theorem~\ref{thm:1F-CSP}.
The starting point is that
given a labeling \(s\) of \((G, w, \pi)\)
a clause \(C(v, u_{1}, \dotsc, u_{t}, x, S)\)
is satisfied if the edges \(\{v, u_{1}\}\), \dots, \(\{v, u_{t}\}\)
are correctly labeled, \(s(v) \in S\), and \(x_{s(v)} \neq 0\):
\begin{equation}
  C(v, u_{1}, \dotsc, u_{t}, x, S) [s^{*}] \geq
  \chi[s(v) = \pi_{v,u_{i}}(s(u_{i})), \forall i \in [t];
  x_{s(v)} \neq 0;
  s(v) \in S]
  .
\end{equation}
Fixing the vertices \(v\), \(u_{1}\), \dots, \(u_{k}\)
and taking expectation over \(x\) and \(S\) yields:
\begin{equation}
  \expectation(\VAR{x}, \VAR{S}){
    C(v, u_{1}, \dotsc, u_{t}, \VAR{x}, \VAR{S}) [s^{*}]}
  \geq
  (1 - \varepsilon) (1 - 1/q) \chi[s(v) = \pi_{v,u_{i}}(s(u_{i})),
  \forall i \in [t]]
  ,
\end{equation}
where the term \((1 - 1/q)\) arises as the probability of
\(\VAR{x}_{s(v)} \neq 0\).
The rest of the proof is identical to that of Theorem
\ref{thm:1F-CSP},
with \(1 - \varepsilon\) replaced with
\((1 - \varepsilon) (1 - 1/q)\).
\end{proof}
\end{theorem}

\section{A small uniform LP over graphs with bounded
  treewidth}
\label{sec:small-uniform-lp}

Complementing the results from before, we now present a
Sherali–Adams like \emph{uniform} LP formulation that solves
\Matching, \IndependentSet, and \VertexCover over graphs of bounded
treewidth.  The linear program has size roughly \(O(n^k)\),
where \(n\)
is the number of vertices and \(k\)
is the upper bound on treewidth.  Here uniform means that
\emph{the same linear program} is used for all graphs of bounded
treewidth with the same number of vertices, in particular, the graph
and weighting are encoded solely in the objective function we
optimize.  This complements recent work \cite{Kolman15}, which
provides a linear program of linear size for \emph{a fixed graph} for
weighted versions of problems expressible in monadic second order
logic. Our approach is also in some sense complementary to
\cite{bienstock2015lp} where small approximate LP formulations are
obtained for problems where the intersection graph of the constraints
has bounded treewidth; here the underlying graph of the problem is of
bounded treewidth. 

Bounded treewidth graphs are of interest, as many NP-hard problems
can be solved in polynomial time when restricting to graphs of bounded
treewidth.  The celebrated Courcelle's Theorem \cite{courcelle90}
states that
any graph property definable by a monadic second order formula
can be decided for bounded treewidth graphs
in time linear in the size of the graph
(but not necessarily polynomial in the treewidth or the size of the
formula).

The usual approach to problems for graphs of bounded treewidth
is to use dynamic programming to select and patch together
the best partial solutions defined on small portions of the graph.
Here we model this in a linear program, with the unique feature
that it does not depend on any actual tree decomposition.
We call problems \emph{admissible} which have the necessary additional
structure, treating partial solutions and restrictions
in an abstract way.

\begin{definition}[Admissible problems]
 Let \(n\) and \(k\) be positive integers.
 Let \(\mathcal{P} = (\mathcal{S}, \mathfrak{G}_{n, k}, \val)\) be an
 optimization problem
 with instances the set \(\mathfrak{G}_{n, k}\) of all graphs \(G\)
 with \(V(G) \subseteq [n]\) and \(\tw(G) \leq k\).
 The problem \(\mathcal{P}\) is \emph{admissible} if
 \begin{enumerate}
 \item \emph{Partial feasible solutions.}
   There is a set \(\mathcal{S} \subseteq \mathscr{S}\) of
   \emph{partial feasible solutions}
   and a restriction operation \(\restriction\)
   mapping any partial solution \(s\) and a vertex set \(X \subseteq
   [n]\) to a partial solution \(s \restriction X\).
   We assume the identity
   \((s \restriction X) \restriction Y = s \restriction Y\)
   for all vertex sets \(X, Y \subseteq V(G)\)
   with \(Y \subseteq X\)
   and partial solutions \(s \in \mathscr{S}\).
   Let
   \(\mathscr{S}_{X} \coloneqq \set{s \restriction X}{s \in \mathcal{S}}\)
   denote the set of restriction of all feasible solutions to \(X\).
 \item \emph{Locality.}
   The measure \(\val_{G}(s)\) depends only on \(G\)
   and \(s \restriction V (G)\) for a graph
   \(G \in \mathfrak{G}_{n,k}\) and a solution \(s \in \mathcal{S}\).
 \item \emph{Gluing.}
   For any cover \(V(G) = V_{1} \cup \dotsb \cup V_{l}\)
   satisfying \(E[G] = E[V_{1}] \cup \dotsb \cup E[V_{l}]\)
   and any feasible solutions \(\sigma_{1} \in \mathscr{S}_{V_{1}}\),
   \dots, \(\sigma_{l} \in \mathscr{S}_{V_{l}}\)
   satisfying
   \begin{equation}
     \sigma_{i} \restriction V_{i} \cap V_{j}
     =
     \sigma_{j} \restriction V_{i} \cap V_{j}
     \qquad \text{for all \(i \neq j\)}
     ,
   \end{equation}
   there is a unique feasible solution \(s\)
   with \(s \restriction V_{i} = \sigma_{i}\) for all \(i\).
\item \emph{Decomposition.}
Let \(T\) be an arbitrary tree decomposition of a graph \(G\)
with \(\tw (G) \leq k\) with bags \(B_{v}\) at nodes \(v \in T\).
Let \(t \in V(T)\) be an arbitrary node of \(T\).
Let \(T_{1}, \dotsc, T_{m}\)
be the components of \(T \setminus t\)
and \(t_{i} \in V(T_{i})\) be the unique node \(t_{i}\)
in \(T\) connected to \(t\). Clearly, every \(T_{i}\) is a tree decomposition
of an induced subgraph \(G_{i} = G[\bigcup_{p \in T_{i}} B_{p}]\)
of \(G\). Moreover, \(B_{t} \cap V (G_{i}) = B_{t} \cap B_{t_{i}}\).

We require the existence of a (not necessarily nonnegative) function
\(\corr_{G, T, t}\) such that for all feasible solution \(s\)
\begin{equation}
  \label{eq:tree-split}
  \val_{G}(s)
  = \corr_{G, T, t}(s \restriction B_{t})
  +
  \sum_{i}
  \val_{G_{i}}(s)
  .
\end{equation}
 \end{enumerate}
\end{definition}
The decomposition property forms the basis of the mentioned dynamic
approach, which together with the gluing property allows
the solutions to be built up
from the best compatible pieces.
The role of the locality property is to ensure that the value function
is independent of irrelevant parts of the feasible
solutions.
In particular, \eqref{eq:tree-split} generalizes for the optima,
when the restriction \(\sigma\) of the solution to \(B_{t}\) is fixed,
this is also the basis of the dynamic programming approach mentioned
earlier:
\begin{lemma}
  \label{lem:decomposition-OPT}
  For any admissible problem \(\mathcal{P}\),
  with the assumption and notation of the decomposition property
  we have for any \(\sigma \in \mathscr{S}_{B_{t}}\)
  \begin{equation}
    \label{eq:tree-split-max}
    \OPT[s \colon s \restriction B_{t} = \sigma]
    {\val_{G}(s)}
    = \corr_{G, T, t}(\sigma)
    +
    \sum_{i}
    \OPT[s \colon s \restriction B_{t_i} \cap B_{t}
    = \sigma \restriction B_{t_i} \cap B_{t}]
    {\val_{G_{i}}(s)}
    .
  \end{equation}
\end{lemma}
\begin{proof}
For simplicity, we prove this only for maximization problems,
as the proof for minimization problems is similar.
By \eqref{eq:tree-split},
the left-hand side is clearly less than or
equal to the right-hand side.  To show equality let
\[s_{i} \coloneqq \argmax_{s \colon s \restriction B_{t} \cap
  B_{t_{i}} = \sigma \restriction B_{t} \cap B_{t_{i}}} \val_{G_{i}}(s)\]
be maximizers.
We apply the gluing property for the
\(s_{i} \restriction V(G_{i})\) and \(\sigma\).

First we check that the conditions for the property are satisfied.
By the properties of a tree decomposition,
we have \(V(G) = B_{t} \cup \bigcup_{i} V(G_{i})\)
and \(E[V(G)] = E[B_{t}] \cup \bigcup_{i} E[V(G_{i})]\).
Moreover, \(B_{t} \cap V(G_{i}) = B_{t} \cap B_{t_{i}}\),
and hence
 \(s_{i} \restriction (B_{t} \cap V(G_{i}))
= \sigma \restriction (B_{t} \cap V(G_{i}))\).
Again by the properties of tree decomposition,
for \(i \neq j\),
it holds \(V(G_{i}) \cap V(G_{j}) \subseteq B_{t}\), and hence
\begin{equation}
 \begin{split}
   s_{i} \restriction (V(G_{i}) \cap V(G_{j}))
   &
   =
   s_{i} \restriction (B_{t} \cap (V(G_{i})) \cap V(G_{j}))
   \\
   &
   =
   \sigma \restriction (B_{t} \cap V(G_{i})) \cap V(G_{j})
   =
   \sigma \restriction (V(G_{i}) \cap V(G_{j}))
   .
 \end{split}
\end{equation}
In particular, \(s_{i} \restriction (V(G_{i}) \cap V(G_{j})) = s_{j}
\restriction (V(G_{i}) \cap V(G_{j}))\).

Therefore by the gluing property,
there is a unique feasible solution \(s\)
with \(s \restriction B_{t} = \sigma\)
and \(s \restriction V(G_{i}) = s_{i} \restriction V(G_{i})\)
for all \(i\).
Clearly, \(\val_{G} (s)\) is equal to the right-hand side.
\end{proof}

We are ready to state the main result of this section,
the existence of a small linear programming formulation
for bounded treewidth graph problems:
\begin{theorem}[Uniform local LP formulation]
  \label{thm:uniformCase}
  Let \(\mathcal{P} = (\mathcal{S}, \mathfrak{G}_{n, k}, \val)\) be an
  admissible optimization problem.
  Then it has the following linear programming formulation,
  which does not depend on any tree decomposition of the instance
  graphs, and has size
  \begin{equation}
    \label{eq:treewidth-factorization-size}
    \fc(\mathcal{P})
    \leq
    \sum_{X \subseteq V(G), \size{X} < k}
    \size{\mathscr{S}_{X}}
    .
  \end{equation}
  The guarantees are \(C(G) = S(G) =\OPT{G}\).
  Let \(V_{0}\) be the real vector space with coordinates indexed by
  the \(X, \sigma\) for \(X \subseteq V(G)\), \(\sigma \in \mathscr{S}_{X}\)
  with \(\size{X} < k\).
  \begin{description}
  \item[Feasible solutions]
    A feasible solution \(s \in \mathcal{S}\) is represented by
    the vectors \(x^{s}\) in \(V_{0}\) with coordinates
    \(x^{s}_{X, \sigma} \coloneqq \chi(s \restriction X = \sigma)\).
  \item[Domain] The domain of the linear program
    is the affine space \(V\) spanned by all the \(x^{s}\).
  \item[Inequalities]
    The LP has the inequalities \(x \geq 0\).
  \item[Instances]
    An instance \(G\) is represented by the unique affine function
    \(w^{G} \colon V \to \R\) satisfying
    \(w^{G}(x^{s}) = \val_{G}(s)\).
  \end{description}
\end{theorem}
One can eliminate the use of the affine subspace \(V\),
by using some coordinates for \(V\) as variables
for the linear program.

\begin{remark}[Relation to the Sherali–Adams hierarchy]
  The linear program above is inspired by the
  Sherali-Adams hierarchy \cite{SheraliAdams1990} as well as the generalized
  extended formulations model in \cite{BPZ2015}.
  The LP is the standard \((k-1)\)-round Sherali–Adams hierarchy
  when \(\mathcal P\) arises from a CSP:
  the solution set \(\mathcal{S}\) is simply the set of all subsets
  of \(V(G)\), and one chooses \(s \restriction X = s \cap X\).
  The inequalities of the LP are the linearization of
  the following functions, in exactly the same way as for the
  Sherali–Adams hierarchy:
\[\chi (s \restriction X = \sigma) \coloneqq \prod_{i \in \sigma} x_i \prod_{i \in
X \setminus \sigma } (1-x_i).\]
For non-CSPs the local functions take on different meanings that are
incompatible with the Sherali-Adams perspective.
\end{remark}

With this we are ready to prove the main theorem of this section.

\begin{proof}[Proof of Theorem~\ref{thm:uniformCase}]
We shall prove that
there is a nonnegative factorization
of the slack matrix of \(\mathcal{P}\)
\begin{equation}
  \label{eq:treewidth-factorization}
  \tau [\OPT{G} - \val_{G}(s)] =
  \sum_{\substack{X \subseteq V(G), \size{X} < k \\
      \sigma \in \mathscr{S}_{X}}}
  \alpha_{G, X, \sigma} \cdot
  \chi (s \restriction X = \sigma)
  ,
\end{equation}
where \(\tau = 1\) if \(\mathcal{P}\) is a maximization problem,
and \(\tau = -1\) if it is a minimization problem.

From this, one can define the function \(w^{G}\) as:
\begin{equation}
  w^{G}(x) \coloneqq
  \OPT{G} -
  \tau^{-1} \sum_{\substack{X \subseteq V(G), \size{X} < k \\
      \sigma \in \mathscr{S}_{X}}}
  \alpha_{G, X, \sigma} \cdot
  x_{X, \sigma}
  ,
\end{equation}
such that it is immediate that \(w^{G}\) is affine,
\(w^{G}(x^{s}) = \val_{G}(s)\)
for all \(s \in \mathcal{S}\),
and that \(\tau [\OPT{G} - w^{G}(x)] \geq 0\)
for all \(x \in V\) satisfying the LP inequalities \(x \geq 0\).
The uniqueness of the \(w^{G}\) follow from \(V\)
being the affine span of the points \(x^{s}\),
where \(w^{G}\) has a prescribed value.

To show \eqref{eq:treewidth-factorization},
let us use the setup for the decomposition property:
Let \(t\) be a node of \(T\),
and let \(t_{1}\), \dots, \(t_{m}\) be the neighbors of \(t\),
and \(T_{i}\) be the component of \(T \setminus t\) containing
\(t_{i}\).
Let \(B_{x}\) denote the bag of a node \(x\) of \(T\).
Let \(G_{i} \coloneqq G[\bigcup_{p \in T_{i}} B_{p}]\) be the induced
subgraph of \(G\) for which \(T_{i}\) is a tree decomposition
(with bags inherited from \(T\)).

We shall inductively define nonnegative numbers
\(\alpha_{G, X, \sigma, A}\) for \(G \in \mathfrak{G}_{n, k}\),
\(X \subseteq V(G)\), \(\sigma \in \mathscr{S}_{X}\),
and \(A \subseteq B_{t}\)
satisfying
\begin{equation}
  \label{eq:uniformCase-induction}
  \tau \left[
    \OPT[s' \colon s' \restriction A = s \restriction A]
    {\val_{G}(s')}
    - \val_{G}(s)
  \right]
  =
  \sum_{X \subseteq V(G), \sigma \in \mathscr{S}_{X}}
  \alpha_{G, X, \sigma, A}
  \cdot
  \chi(s \restriction X = \sigma)
  .
\end{equation}
This will prove the claimed \eqref{eq:treewidth-factorization}
with the choice
\(\alpha_{G, X, \sigma} \coloneqq \alpha_{G, X, \sigma, B_{t}}\).
The help variable \(A\) is only for the induction.

To proceed with the induction,
we take the difference of
Eqs.~\eqref{eq:tree-split} and \eqref{eq:tree-split-max}
with the choice \(\sigma \coloneqq s \restriction B_{t}\):
\begin{equation}
  \tau
  \left[
    \OPT[s' \colon s' \restriction B_{t} = \sigma]
    {\val_{G}(s')}
    -
    \val_{G}(s)
  \right]
  =
  \sum_{i}
  \tau
  \left[
    \OPT[s' \colon s' \restriction B_{t_i} \cap B_{t}
    = \sigma \restriction B_{t_i} \cap B_{t}]
    {\val_{G_{i}}(s')}
    -
    \val_{G_{i}}(s)
  \right]
  .
\end{equation}
Now we use the induction hypothesis on the \(G_{i}\)
with tree decomposition \(T_{i}\) to obtain
\begin{equation}
  \tau
  \left[
    \OPT[s' \colon s' \restriction B_{t} = \sigma]
    {\val_{G}(s')}
    -
    \val_{G}(s)
  \right]
  =
  \sum_{i}
  \sum_{\substack{X \subseteq V(G) \\ \size{X} < k \\
      \sigma \in \mathscr{S}_{X}}}
  \alpha_{G_{i}, X, \sigma, B_{t_{i}} \cap B_{t}}
  \chi(s \restriction X = \sigma)
  .
\end{equation}
Hence \eqref{eq:treewidth-factorization} follows with the
following choice of the \(\alpha_{G, X, \sigma, A}\),
which are clearly nonnegative:
\begin{equation}
  \alpha_{G, X, \sigma, A}
  \coloneqq
  \begin{dcases}
    \sum_{i}
    \alpha_{G_{i}, X, \sigma, B_{t_{i}} \cap B_{t}}
    &
  \text{if } X \neq B_{t}
  \\
  \sum_{i}
  \alpha_{G_{i}, X, \sigma, B_{t_{i}} \cap B_{t}}
  + \tau
  \left[
    \OPT[s' \colon s' \restriction A = \sigma \restriction A]
    {\val_{G}(s')}
    -
    \OPT[s' \colon s' \restriction B_{t} = \sigma]
    {\val_{G}(s')}
  \right]
  &
  \text{if } X = B_{t}
  .
  \end{dcases}
  \qedhere
\end{equation}
\end{proof}

We now demonstrate the use of Theorem~\ref{thm:uniformCase}.

\begin{example}[\VertexCover,
  \IndependentSet, and CSPs such as e.g., \MaxCUT, \UniqueGames]
  For the problems \MaxCUT, \IndependentSet, and \VertexCover,
  the set of feasible solutions \(\mathcal{S}\) is the set of all
  subsets of \(\mathcal{S}\).
  We need no further partial solutions
  (i.e., \(\mathscr{S} \coloneqq \mathcal{S}\)), and
  we choose the restriction to be simply the intersection
\[s \restriction B \coloneqq s \cap B.\]
  It is easily seen that this makes \IndependentSet and \VertexCover
  admissible problems, providing an LP of size \(O(n^{k-1})\)
  for graphs with treewidth at most \(k\).
  As an example, we check the decomposition property
  for \IndependentSet.
  Using the same notation as in the decomposition property,
  \begin{equation}
   \begin{split}
    \val_{G}(s)
    -
    \sum_{i}
    \val_{G_{i}}(s)
    &
    =
    \size{s} - \sum_{i}
    \left\{
      \size{s \cap V(G_{i})} - \size{E(G_{i}[s \cap V(G_{i})])}
      \strut
    \right\}
    \\
    &
    =
    \size{s \cap B_{t}} - \sum_{i}
    \left\{
      \size{s \cap B_{t} \cap V(G_{i})} -
      \size{E(G_{i}[s \cap B_{t} \cap V(G_{i})])}
      \strut
    \right\}
    ,
   \end{split}
  \end{equation}
  as any vertex \(v \notin B_{t}\) is a vertex of exactly
  one of the \(G_{i}\), and similarly for edges
  with at least one end point not in \(B_{t}\).
  Therefore the decomposition property is satisfied with
  the choice
  \begin{equation}
    \corr_{G, T, t}(\sigma)
    \coloneqq
    \size{\sigma \cap B_{t}} - \sum_{i}
    \left\{
      \size{\sigma \cap B_{t} \cap V(G_{i})} -
      \size{E(G_{i}[\sigma \cap B_{t} \cap V(G_{i})])}
      \strut
    \right\}
    .
  \end{equation}

  For \UniqueGames{n}{q}, the feasible solutions are all functions
  \([n] \to [q]\).
  Partial solutions are functions \(X \to [q]\)
  defined on some subset \(X \subseteq [n]\).
  Restriction \(s \restriction X\)
  is the usual restriction of \(s\) to the subset \(\dom(s) \cap X\).
  This obviously makes \MaxCUT and \UniqueGames{n}{q} admissible.
  The size of the LP is \(O(n^{2 (k-1)})\) for \MaxCUT,
  and \(O((q n^{2})^{k-1})\) for \UniqueGames{n}{q}.
\end{example}

The \Matching problem requires
that the restriction operator preserves
more local information to ensure that partial solutions
are incompatible when they contain a different edge
at the same vertex.

\begin{example}[\Matching]
  The \Matching problem has feasible solutions all perfect
  matchings.
  The partial solutions are all matchings, not necessarily perfect.
  The restriction \(s \restriction X\) of a matching \(s\) to a vertex
  set \(X\) is defined as the set of all edges in \(s\) incident to
  some vertex of \(X\):
  \[s \restriction X \coloneqq
  \set{\{u,v\} \in s}{u \in B \lor v \in B}
  .\]
  Now \(s \restriction X\) can contain edges
  with only one end point in \(X\).
  Again, this makes \Matching an admissible problem,
  providing an LP of size \(O(n^{k})\) (the number of edges with at
  most \(k\) edges).
  Here we check the gluing property.
  Let \(V(G) = V_{1} \cup \dotsb \cup V_{l}\) be a covering
  (we do not need
  \(E[V(G)] = E[V(G_{1})] \cup \dotsb \cup E[V(G_{l})]\)),
  and let \(\sigma_{i}\) be a (partial) matching
  covering \(V_{i}\) with every edge in \(\sigma_{i}\)
  incident to some vertex in \(V_{i}\)
  (i.e., \(\sigma_{i} \in \mathscr{S}_{V_{i}}\)) for \(i \in [l]\).
  Let us assume \(\sigma_{i} \restriction V_{i} \cap V_{j} =
  \sigma_{j} \restriction V_{i} \cap V_{j}\),
  i.e., every vertex \(v \in V_{i} \cap V_{j}\)
  is matched to the same vertex by \(\sigma_{i}\) and \(\sigma_{j}\)
  for \(i \neq j\).
  It readily follows that the union \(s \coloneqq \bigcup_{i} \sigma_{i}\)
  is a matching.
  Actually, it is a perfect matching as
  \(V(G) = V_{1} \cup \dotsb \cup V_{l}\)
  ensures that it covers every vertex.
  Obviously, \(s \restriction V_{i} = \sigma_{i}\)
  and \(s\) is the unique perfect matching with this property.
\end{example}

\section*{Acknowledgements}
\label{sec:acknowledgements}

Research reported in this paper was partially supported by NSF CAREER
award CMMI-1452463. Parts of this research was conducted at the CMO-BIRS
2015 workshop \emph{Modern Techniques in Discrete Optimization: Mathematics,
Algorithms and Applications} and we would like to thank the organizers
for providing a stimulating research environment,
as well as Levent Tunçel for helpful discussions on
Lasserre relaxations of the \IndependentSet problem.

\bibliographystyle{alphaabbr}
\bibliography{bibliography}
\end{document}